%% file: main.tex
\documentclass[twoside]{article}

\usepackage[accepted]{AISTATS2024/aistats2024}
\input{preamble}

\usepackage[round]{natbib}
\begin{document}

%

%

\twocolumn[

\aistatstitle{VEC-SBM: Optimal Community Detection with Vectorial Edges Covariates}

\aistatsauthor{Guillaume Braun \And Masashi Sugiyama }

\aistatsaddress{ RIKEN AIP \And RIKEN AIP \\ University of Tokyo} ]

\begin{abstract}
 Social networks are often associated with rich side information, such as texts and images. While numerous methods have been developed to identify communities from pairwise interactions, they usually ignore such side information. In this work, we study an extension of the Stochastic Block Model (SBM), a widely used statistical framework for community detection, that integrates vectorial edges covariates: the Vectorial Edges Covariates Stochastic Block Model (VEC-SBM). We propose a novel algorithm based on iterative refinement techniques and show that it optimally recovers the latent communities under the VEC-SBM. Furthermore, we rigorously assess the added value of leveraging edge's side information in the community detection process. We complement our theoretical results with numerical experiments on synthetic and semi-synthetic data.
\end{abstract}

\input{intro}
\input{model}
\input{main_results}

\input{num_xp}

\section{Conclusion and perspectives}
 We have quantified the added value of edge-side information within the VEC-SBM framework. Our findings reveal that incorporating edge covariates can significantly improve the SNR, particularly when communities exhibit similar connectivity profiles or when dealing with a large number of communities. Furthermore, we have introduced an efficient iterative algorithm, \sir, which has been proven to achieve the optimal misclustering rate. 
 
However, our work leaves several questions open for future research. These include the challenging task of estimating the number of communities in the presence of covariates, as well as the analysis of random initialization techniques for improved community detection. Furthermore, a promising research direction is to extend this framework to more complex and realistic models where the covariates are high-dimensional, and the network possesses intricate structures beyond the scope of the traditional SBM.

\subsubsection*{Acknowledgements} G.B. would like to express gratitude to the anonymous reviewers for their constructive feedback, contributing to the overall clarity and quality of the paper. Special thanks to Okan Koc for his valuable assistance concerning the Python implementation. M.S. was supported by JST CREST Grant Number JPMJCR18A2 and a grant from Apple, Inc. Any views, opinions, findings, and conclusions or recommendations expressed in this material are those of the authors and should not be interpreted as reflecting the views, policies or position, either expressed or implied, of Apple Inc.


\bibliography{references}

\section*{Checklist}

 \begin{enumerate}
 \item For all models and algorithms presented, check if you include:
 \begin{enumerate}
   \item A clear description of the mathematical setting, assumptions, algorithm, and/or model. [Yes]
   \item An analysis of the properties and complexity (time, space, sample size) of any algorithm. [Yes]
   \item (Optional) Anonymized source code, with specification of all dependencies, including external libraries. [No]
 \end{enumerate}

 \item For any theoretical claim, check if you include:
 \begin{enumerate}
   \item Statements of the full set of assumptions of all theoretical results. [Yes]
   \item Complete proofs of all theoretical results. [Yes]
   \item Clear explanations of any assumptions. [Yes]     
 \end{enumerate}

 \item For all figures and tables that present empirical results, check if you include:
 \begin{enumerate}
   \item The code, data, and instructions needed to reproduce the main experimental results (either in the supplemental material or as a URL). [Yes]
   \item All the training details (e.g., data splits, hyperparameters, how they were chosen). [Yes]
         \item A clear definition of the specific measure or statistics and error bars (e.g., with respect to the random seed after running experiments multiple times). [Yes]
         \item A description of the computing infrastructure used. (e.g., type of GPUs, internal cluster, or cloud provider). [Yes]
 \end{enumerate}

 \item If you are using existing assets (e.g., code, data, models) or curating/releasing new assets, check if you include:
 \begin{enumerate}
   \item Citations of the creator If your work uses existing assets. [Yes]
   \item The license information of the assets, if applicable. [Not Applicable]
   \item New assets either in the supplemental material or as a URL, if applicable. [Not Applicable]
   \item Information about consent from data providers/curators. [Not Applicable]
   \item Discussion of sensible content if applicable, e.g., personally identifiable information or offensive content. [Not Applicable]
 \end{enumerate}

 \item If you used crowdsourcing or conducted research with human subjects, check if you include:
 \begin{enumerate}
   \item The full text of instructions given to participants and screenshots. [Not Applicable]
   \item Descriptions of potential participant risks, with links to Institutional Review Board (IRB) approvals if applicable. [Not Applicable]
   \item The estimated hourly wage paid to participants and the total amount spent on participant compensation. [Not Applicable]
 \end{enumerate}

 \end{enumerate}
\newpage
\appendix
\input{appendix}

\end{document}

%% file: preamble.tex
\usepackage[utf8]{inputenc}
\usepackage[english]{babel}
\usepackage{amsmath,amssymb,epsfig,amsthm}
\usepackage{graphicx}
\usepackage{comment}
\usepackage{array}
\usepackage{algorithm}
\usepackage{enumerate}
\usepackage{xcolor}
\usepackage{algpseudocode}
\usepackage{xurl}
\usepackage{hyperref}

\allowdisplaybreaks[1]

\newtheorem{theorem}{Theorem}

\newtheorem{assumption}{Assumption}

\newtheorem{condition}{Condition}

\newtheorem{corollary}{Corollary}

\newtheorem{lemma}{Lemma}

\newtheorem{remark}{Remark}

\numberwithin{equation}{section}

\DeclareMathOperator{\Var}{Var}

\newcommand{\calB}{\ensuremath{\mathcal{B}}}
\newcommand{\calC}{\ensuremath{\mathcal{C}}}
\newcommand{\calD}{\ensuremath{\mathcal{D}}}

\newcommand{\calN}{\ensuremath{\mathcal{N}}}

\newcommand{\calE}{\ensuremath{\mathcal{E}}}

\newcommand{\calO}{\ensuremath{\mathcal{O}}}

\newcommand{\norm}[1]{\left\|{#1}\right\|}
\newcommand{\abs}[1]{\left|{#1}\right|}

\newcommand{\expec}{\ensuremath{\mathbb{E}}}

\newcommand{\prob}{\ensuremath{\mathbb{P}}}

\newcommand{\lt}{l(z^{(t)},z)}

\definecolor{asparagus}{rgb}{0.53, 0.66, 0.42}

\newcommand{\gb}[1]{\textcolor{cyan}{#1}}

\newcommand{\indic}{\ensuremath{\mathbf{1}}} 
\newcommand*\diff{\mathop{}\!\mathrm{d}}

\newcommand{\R}{\ensuremath{\mathbb{R}}}

\newcommand{\spec}{\texttt{Spec}}

\newcommand{\olmf}{\texttt{OLMF}}
\newcommand{\ir}{\texttt{IR-VEC}}
\newcommand{\sir}{\texttt{sIR-VEC}}

%% file: intro.tex
\section{Introduction}
Networks are a powerful tool for representing relational data, where each entity is represented by a node and pairwise connections between these entities are encoded by edges. Over the past decades, numerous clustering methods have been devised to extract meaningful insights from graph-structured datasets. These methods are often evaluated under the Stochastic Block Model \citep{HOLLAND1983109}, a random graph model where each edge is sampled independently with a probability depending solely on the latent communities of the corresponding nodes. However, a notable limitation of the SBM is its exclusive consideration of binary interactions.

In recent years, there has been growing interest in developing extensions of the SBM that can incorporate more information. This includes the Weighted SBM \citep{wsbm2014}, which assigns scalar weights to edges, the Labeled SBM \citep{lSBM} where labels correlated to the nodes' community are available, the Multi-Layer SBM \citep{mlsbm}, accommodating multimodal interactions in distinct layers, the Contextual SBM \citep{csbm}, linking each node to a covariate vector, or the recently introduced Embedded Topic SBM \citep{boutin_embedded_2023}, where textual information is associated with each edge.

In this work, we consider a variant of the Embedded Topic SBM: the Vectorial Edges Covariate SBM (VEC-SBM). Under the VEC-SBM, each observed edge is associated with a vector, distinguishing it from the Weighted SBM, which solely permits scalar weights, and the Multi-Layer SBM, where edge presence is sampled independently on each layer. Consequently, methodologies and theoretical guarantees established for these existing models cannot be directly applied to the VEC-SBM. Moreover, in contrast to prior work such as that by \cite{boutin_embedded_2023}, which primarily focuses on practical applications, our study centers on the statistical analysis of the VEC-SBM. In particular, our analysis quantifies the added value of the side information provided by the edges and shows that it has a multiplicative effect on the signal-to-noise ratio (SNR). This is in contrast with the node's side information that has an additive impact on the SNR \citep{abbe_lp}, further motivating the incorporation of edges' side information in clustering algorithms.

\paragraph{Our contributions.} In this work, we make the following contributions.
\begin{itemize}
    \item We introduce a novel algorithm for graph clustering that incorporates edge vectorial covariates. Our algorithm is computationally efficient and applicable to various settings. 
    \item We rigorously analyze our algorithm under the VEC-SBM and demonstrate that it achieves a statistically optimal convergence rate. We also provide valuable insights by quantifying the information gain resulting from the inclusion of edge covariates: under the VEC-SBM, the SNR will depend on the difference between the means of the covariates of different classes multiplied by the average degree of the nodes in the graph. Even if this difference of means is small this can lead to a considerable increase in the SNR depending on the sparsity level of the graph. 
    \item  We conduct comprehensive numerical experiments on synthetic data to further substantiate our findings. These experiments highlight the importance of leveraging both the network structure and covariate information for effectively recovering all communities. Moreover, we apply our algorithm to a real-world dataset, with synthetic edge covariates, demonstrating its practical applicability.
\end{itemize}

\paragraph{Related Work.} Various extensions of the SBM have been developed to incorporate side information. In this paragraph, we describe the main existing algorithmic approaches used for this purpose.
\begin{itemize}
    \item \textit{Discretization of edge weights.} \cite{Xu2017OptimalRF} introduced an algorithm that achieves the optimal rate of convergence under the Weighted SBM by discretizing edge weights. However, this strategy becomes computationally inefficient when dealing with high-dimensional covariates, and selecting the appropriate level of discretization can be challenging.
    \item \textit{Tensor and matrix factorization methods.} Tensor methods, such as those presented in \cite{twist} or \cite{paul2020}, offer an alternative approach to incorporate side information. However, these methods are typically applied to multi-layer graphs, which exhibit a distinct noise structure compared to our setting. 
    \item \textit{Model-based approaches.} Since the Maximum Likelihood Estimator (MLE) is intractable for models based on the SBM, several alternative approaches have been used. \cite{cerqueira2023pseudolikelihood} proposed a pseudo-likelihood approach that avoids the need for discretization, but it relies on the assumption that weights are sampled from a univariate Gaussian mixture. Their method attains a convergence rate matching the lower bound established by Xu et al. \cite{Xu2017OptimalRF} under the Weighted SBM, albeit up to a constant factor. Another prevalent model-based approach involves the use of variational EM algorithms \citep{r2016Blockmodels, bouveyron2016}. While effective, these methods are often computationally demanding and lack theoretical guarantees. The recent work of \cite{boutin_embedded_2023, boutin2023deep} combined the variational approach with a deep learning architecture to simultaneously extract topics from edges' textual side information and cluster the network. In this work, we consider a simpler setting where the edge covariates can be directly exploited. Our main purpose is to statistically analyze the model to get insight into the added value of edge-side information.
    \item \textit{Iterative refinement approaches/ alternating optimization.}
    Instead of using the pseudolikelihood or the variational method, our algorithm relies on an iterative refinement procedure based on a simplified version of the Maximum A Posteriori (MAP) estimator. The global proof strategy is based on the seminal work of \cite{Gao2019IterativeAF} and inspired by the analysis of the Contextual SBM (CSBM) \citep{csbm_braun22a}. However, contrary to the CSBM, there is a dependency between the two available sources of information under the VEC-SBM, since we only access to edge covariates for the observed edges. This poses new challenges and requires the development of new analysis techniques. Due to the particular noise structure in our setting, the error decomposition is different and requires new concentration inequalities and techniques to be controlled.
    Using alternating optimization to solve non-convex optimization problems is a popular approach that has also been used for other problems including matrix sensing \citep{Stger2021SmallRI}, dictionary learning \citep{liang2022simple}, heterogeneous matrix factorization \citep{shi2023heterogeneous}, and multi-task regression \citep{metalearning} to mention but a few. 
\end{itemize}

\paragraph{Notations.}
 We will use Landau standard notations $o(.)$ and $O(.)$. For sequences $(a_n)_{n \geq 1}$ and $(b_n)_{n \geq 1}$, if there is a constant $C>0$ such that $a_n \leq C b_n$ (resp. $a_n \geq C b_n$) for all $n$ we will write $a_n\lesssim b_n$ (resp. $a_n\gtrsim b_n$ ). If $a_n \lesssim b_n$ and $a_n \gtrsim b_n$, then we write $a_n \asymp b_n$.  Matrices will be denoted by uppercase letters. The $i$-th row of a matrix $A$ will be denoted as $A_{i:}$ and depending on the context can be interpreted as a column vector. The column $j$ of $A$ will be denoted by $A_{:j}$, and the $(i,j)$th entry by $A_{ij}$. The transpose of $A$ is denoted by $A^\top$. $I_k$ denotes the $k\times k$ identity matrix. We use $\|\cdot\|$ and $\|\cdot\|_F$ to respectively denote the spectral norm (or Euclidean norm in the case of vectors) and the Frobenius norm.  The number of non-zero entries of a matrix $A$ is denoted $\mathrm{nnz}(A)$. $A^\dagger$ denote the pseudo-inverse of $A$. The maximum between $a$ and $b$ will be denoted by $a\vee b$. The indicator function of a set $C$ is denoted by $\indic_C$. 

%% file: model.tex
\section{Model and algorithm description} In Section \ref{subsec:stat_frame} we describe the generative model and discuss the assumptions we made for the analysis. Then, in Section \ref{subsec:algo} we introduce our algorithm.
\subsection{The Vectorial Edge Covariates Stochastic Block Model (VEC-SBM)}
\label{subsec:stat_frame}

The VEC-SBM is an extension of the SBM where a vector is associated with each observed edge of a graph sampled from an SBM. The distribution of these edge vector covariates only depends on the community the endpoints forming the edge belong to.  More formally, the VEC-SBM is defined by the following parameters.
\begin{itemize}
    \item A set of nodes $\calN = [n]$ and a partition of $\calN$ into $K$ communities $\calC_1,\ldots, \calC_K$. 
    
    \item A membership matrix $Z \in \lbrace 0,1\rbrace^{n\times K}$ such that there is exaclty one $1$ in every row.  Each membership matrix $Z$ can be associated bijectively with a partition function $z:[n]\to [K]$ such that $z(i)=z_i=k$  where $k$ is the unique column index satisfying $Z_{ik}=1$.

    \item A symmetric connectivity matrix of probabilities between communities $$\Pi=(\Pi_{k k'})_{k, k'\in [K]} \in [0,1]^{K \times K}.$$ 

    \item A family of centroids $\mu=(\mu_{kk'})_{k,k'\in [K]}$ such that $\mu_{kk'}\in \R^d$, where $d=O(1)$ and $\mu_{kk'}=\mu_{k'k}$ for all $k, k'\in [K]$.

    \item A family of covariance matrices $\Sigma =(\Sigma_{kk'})_{k, k'\in [K]}$ such that $\Sigma_{kk'}\in \R^{d\times d}$ and $\Sigma_{kk'}=\Sigma_{k'k}$ for all $k, k'\in [K]$.
\end{itemize}

 A graph with $n$ nodes and edge covariates of dimension $d$ can be represented by a tensor $R\in \R^{n\times n \times d}$. It is distributed according to the VEC-SBM$(Z, \Pi, \mu, \Sigma)$ if the entries of $R$ are generated as follows. First, the presence or absence of an edge is encoded by a matrix $A\in \lbrace 0,1 \rbrace ^{n\times n}$ with entries samples independently by 
\[ 
A_{ij}  \overset{\text{ind.}}{\sim} \mathcal{B}(P_{ij}), \quad i, j \in [n],\, i\leq j
\] 
where $\calB(p)$ denotes a Bernoulli distribution with parameter $p$, and $P=\expec(A)=Z\Pi Z^\top $.  The noise is denoted by $E=A-\expec(A)$.
The sparsity level of the graph is denoted by $p_{max} = \max_{i,j} p_{ij}$. We will focus on the regime where $np_{max}=\Omega(\log n)$ and $np_{max}=o(n)$. 

Then, for each couple $(i,j)$, we sample independently (conditionally on $Z$) an edge covariate $G_{ij}\in \R^d$ such that
\[ 
G_{ij} = \mu_{z(i)z(j)}+ \epsilon_{ij},
\] 
where $\epsilon_{ij}$ is a centered Gaussian\footnote{We made this assumption to simplify the exposition, but we believe that the proof can be extended to Sub-Gaussian r.v.} variable with covariance $\Sigma_{z(i)z(j)}$. Since multiplying all the edges of the graph by a constant will not modify the information contained in the covariates, we can assume that $\norm{\mu_{kk'}}_\infty \leq 1$ for all $k,k'\in [K]$. 
%
Finally, the observed tensor $R$ is given by $R_{ijd}=A_{ij}(G_{ij})_d$. 
For the analysis, we will make some additional assumptions. 
\begin{assumption}[Approximately balanced communities]\label{ass:balanced_part}
The communities  $\calC_1,\ldots, \calC_K$
 are approximately balanced, i.e., there exists a constant $\alpha \geq 1$ such that for all $k\in [K]$ we have 
\[ 
\frac{n}{\alpha K}\leq \abs{\calC_k} \leq \frac{\alpha n}{K}.
\] 
\end{assumption}
\begin{assumption}[Isotropic variance]\label{ass:iso}
For all $k,k'\in [K]$, $\Sigma_{kk'}=I_d$.
\end{assumption}
\begin{assumption}[Symmetric SBM]\label{ass:sym}
The connectivity matrix $\Pi$ is of the form $q\indic\indic^\top+(p-q)I_K$ with $1>p>q>0$.
\end{assumption}

\begin{assumption}[Limited graph information]\label{ass:lim_inf} We have $p=o(1)$, $p/q=O(1)$, $np=\Omega(\log n)$ and $n(\sqrt{p}-\sqrt{q})^2/K<\log n$.
\end{assumption}

Assumption \ref{ass:balanced_part} is standard in clustering literature. When the communities are highly unbalanced, the problem becomes noticeably more difficult, see for example \cite{mukherjee2023recovering}. Assumption \ref{ass:iso} is used to simplify the exposition. We believe that our results can be extended to the general case (e.g. \cite{Chen2021OptimalCI} for Gaussian mixture models) without changing the proof strategy but at the price of additional technicalities. Moreover, our experiments in Section \ref{sec:xp} show that the algorithm we analyzed performs well even if this assumption is not satisfied. Assumption \ref{ass:sym} is also used for convenience, but could be removed at the cost of additional technicalities. Finally, Assumption \ref{ass:lim_inf} implies that there is not enough information in the graph to recover the clusters, motivating the use of side information. In particular, $n(\sqrt{p}-\sqrt{q})^2/K<\log n$ corresponds to the regime where exact recovery is impossible \citep{minimaxSBM}.

 The quality of the clustering is evaluated through the \emph{misclustering rate} $r$ defined by 
 \begin{equation*}
 \label{eq:def_misclust}   
 r(\hat{z},z)=\frac{1}{n}\min _{\pi \in \mathfrak{S}}\sum_{i\in [n]} \indic_{\lbrace \hat{z}(i)\neq \pi(z(i))\rbrace},\end{equation*}
 where $\mathfrak{S}$ denotes the set of permutations on $[K]$. An estimator $\hat{z}$ achieves \emph{exact recovery} if $r(\hat{z},z)=0$ with probability $1-o(1)$ as $n$ tends to infinity. It achieves \emph{weak consistency} (or almost full recovery) if $\prob(r(\hat{z},z)=o(1))=1-o(1)$ as $n$ tends to infinity. A more complete overview of the different types of consistency and the sparsity regimes where they occur can be found in \cite{AbbeSBM}.

\subsection{Algorithm description}
\label{subsec:algo}
 Let us denote \tiny\begin{align*}
     MAP_i&(\calC, \Pi, \mu, \Sigma) = \sum_{l\in [K]}\sum_{j\in \calC_{l}}A_{ij}\log (\Pi_{kl})+(1-A_{ij})\log(1-\Pi_{kl})\\
     &- \sum_{l\in [K]}\sum_{j\in \calC_{l}} A_{ij}(G_{ij}-\mu_{kl})^\top \Sigma_{kl}(G_{ij}-\mu_{kl})-\frac{1}{2}\det(\Sigma_{kl})
 \end{align*} \normalsize the logarithm of the MAP of a node $i$ such that $z(i)=k$, given $\Pi, \mu, \Sigma$ and a partition $\calC$ of $[n]$. At each step $t$, \ir\, (cf. Algorithm \ref{alg:ir}) estimates the model parameters and then updates the partition encoded by $Z^{(t)}$ based on $MAP_i$. We will denote by $\calC_k^{(t)}$ the set of nodes $i$ such that $Z^{(t)}_{ik}=1$, i.e. the nodes that are associated with community $k$ at time $t$.

For the analysis, we will consider a simplified version of \ir\, where at each step $\Sigma_{kk'}=I_d$ for all $k,k'\in [K]$. This version of the algorithm will be referred to as \sir. Despite ignoring the covariance structure, Section \ref{sec:xp} shows that \sir\, performs as well as \ir\, even if the edge covariates are non-isotropic. 

\begin{algorithm}[hbt!]
\caption{Iterative refinement for the VEC-SBM (\ir)}\label{alg:ir}
\begin{flushleft}
        \textbf{Input:} The number of communities $K$, $A$, $G$, an initial estimate of the partition $Z^{(0)}$ of the nodes, a number of iteration $T$.
\end{flushleft}
        \begin{algorithmic}[1]
        \For{$0 \leq t \leq T-1$}  
        \State Estimate the model parameters \tiny \begin{align*}
            W^{(t)} &= (Z^{(t)})^\dagger, \,
            \Pi^{(t)}= W^{(t)^\top}AW^{(t)}\\
            \mu_{kk'}^{(t)} &= \sum_{\substack{i\in \calC_k^{(t)} \\ j\in \calC_{k'}^{(t)}}}A_{ij}G_{ij}/\sum_{\substack{i\in \calC_k^{(t)} \\ j\in \calC_{k'}^{(t)}}}A_{ij}\\
            \Sigma_{kk'}^{(t)}&= \sum_{\substack{i\in \calC_k^{(t)} \\ j\in \calC_{k'}^{(t)}}}A_{ij}(G_{ij}-\mu_{kk'}^{(t)})(G_{ij}-\mu_{kk'}^{(t)})^\top/\sum_{\substack{i\in \calC_k^{(t)} \\ j\in \calC_{k'}^{(t)}}}A_{ij}.
        \end{align*}\normalsize
        \State Update the partition \[ z_i^{(t+1)} = \arg \max_{k\in [K]} MAP_i(\calC^{(t)}, \Pi^{(t)}, \mu^{(t)},\Sigma^{(t)}).\]
       \EndFor
       \end{algorithmic}
        \textbf{Output:} A partition of the nodes $Z^{(T)}$.
\end{algorithm}

\paragraph{Computational complexity.} Estimating the model parameters at each step requires at most $O(\mathrm{nnz}(A)d^2)$ elementary operations. Given the estimate of the model parameters, updating the partition requires $O(\mathrm{nnz}(A)Kd)$ operations where $\mathrm{nnz}(A)$ corresponds to the number of non zero entries of $A$. The global complexity of the algorithm is hence $O(\mathrm{nnz}(A)\max(d^2,Kd))$. Under the VEC-SBM, $\mathrm{nnz}(A)\asymp n^2p_{max}$. 

\paragraph{Initialization.} We use the vanilla spectral method on $A$ for initialization. While the accuracy provided by this method can be very poor in challenging scenarios, we show experimentally in Section \ref{sec:xp} that it doesn't affect the performances of \sir. We leave as an open problem the analysis of random initialization.

%% file: main_results.tex
\section{Analysis method and main results}
In this section, we first introduce some notations associated with the error decomposition, present our main results, and then outline the proof strategy. The details of the proofs can be found in the supplementary material. 

\subsection{Error decomposition}\label{sec:err_dec} Our analysis is based on the framework developed by \cite{Gao2019IterativeAF}. This framework has been used to analyze other clustering models with similar flavors, such as the CSBM \citep{csbm_braun22a} or the Tensor Block Model \citep{tbm}.  However, we emphasize that previous results cannot be adapted straightforwardly to our setting due to the specific noise structure induced by the VEC-SBM. In particular, in the VEC-SBM, there is a dependence between the two sources of noise: the graph and the covariates. This dependence requires new techniques and concentration inequalities to control the noise.

To understand how \sir\, can lead to an improvement of the partition, one needs to analyze in which situation a node $i$ is misclassified after one refinement step. It corresponds to the condition \[ a \neq \arg \max_{k\in [K]} MAP_i(\calC^{(t)}, \Pi^{(t)}, \mu^{(t)},I_d) .\] By some elementary algebra, one can show that the previous condition is equivalent to the existence of $b\in [K]\setminus{\lbrace a \rbrace}$ such that \[ C_i(a,b)< -\Delta^2(a,b)+F_{ib}^{(t)}+G_{ib}^{(t)},\] 
where $F_{ib}^{(t)}$ and $ G_{ib}^{(t)}$ are error terms specified in the appendix (Section \ref{app:err_dec}), and the signal term $\Delta^2(a,b)$ and stochastic term $C_i(a,b)$ are given by \tiny \begin{align}
    \Delta^2(a,b)&= \textcolor{teal}{\log (p/q)(|\calC_a|p\gb{-}|\calC_b|q)}+\textcolor{purple}{\sum_{l\in [K]}|\calC_l|\Pi_{al}\norm{\mu_{al}-\mu_{bl}}^2}\label{eq:delta}\\
    C_i(a,b) &= \log(\frac{p}{q})\left(\sum_{j\in \calC_{a}}E_{ij}-\sum_{j\in \calC_{b}}E_{ij}\right)\notag\\
    &+\sum_{l\in [K]} \sum_{j\in \calC_{l,-i}}\left(E_{ij}\norm{\mu_{al}-\mu_{bl}}^2+2A_{ij}\langle \epsilon_{ij},\mu_{al}-\mu_{bl}\rangle\right)\notag.
\end{align}
\normalsize
The \textcolor{teal}{first part} of the signal only depends on the graph while \textcolor{purple}{the second part} depends on the covariates and the graph connectivity parameters. For instance, in the case where the communities are exactly balanced and $p=q$, \textcolor{purple}{the second part} corresponds to $(np/K)\sum_l\norm{\mu_{al}-\mu_{bl}}^2$, i.e. the sparsity level of the graph $np$ is multiplied by the average distance between the edge covariates means $\sum_l\norm{\mu_{al}-\mu_{bl}}^2/K$, hence the multiplicative effect.

\subsection{Convergence guarantee}
The following theorem shows that if the initialization $z^{(0)}$ is good enough, then \sir\, converges in $O(\log n)$ iterations and achieves a misclustering rate that decreases exponentially in the SNR formally defined as $\Delta^2_{min}=\min_{a\neq b}\Delta^2(a,b)$. 
\begin{theorem}\label{thm:main}  
    Assume that $\Delta_{min}^2 \asymp \log n$. Under assumptions \ref{ass:balanced_part}, \ref{ass:iso},\ref{ass:sym} and \ref{ass:lim_inf}, if $z^{(0)}$ is such that \[ r(z,z^{(0)})\leq  \frac{\epsilon}{K^2}\] for a constant $\epsilon$ small enough, then with probability at least $1-n^{-\Omega(1)}$ we have for all $t\gtrsim \log n $ \[ r(z^{(t)},z)\leq e^{-(c+o(1))\Delta^2_{min}}\] where $c>0$ is the constant appearing in Lemma \ref{lem:conc_oracle}.
\end{theorem}
\begin{remark}
   The condition on initialization implies having $r(z^{(0)},z)=O(1/K^2)$. This is a more stringent requirement compared to the condition in \cite{csbm_braun22a}, which only necessitates $r(z^{(0)},z)=O(1/K)$. This dependency on $K$ is likely to be an artifact of the proof. If we could replace the factor $K^2$ in Lemma \ref{lem:fond_ineq_bis} with $K$, we could relax the initial condition to $r(z^{(0)},z)=O(1/K)$. In Section \ref{sec:xp}, we experimentally demonstrate that \sir\, performs well even when initialized with an almost non-informative $z^{(0)}$.
\end{remark}

\begin{remark} The misclustering rate of the SBM is on the order of $\exp(-n(\sqrt{p}-\sqrt{q})^2/K)$. However, when $p\approx q$  accurate recovery of communities becomes challenging. In a similar context under the VEC-SBM, if $\min_{a \neq b \in [K]}\sum_l\norm{\mu_{al}-\mu_{bl}}^2\geq v>0$, then $\Delta^2_{min}$ is on the order $vnp/K= \Omega(\log n)\gg n(\sqrt{p}-\sqrt{q})^2/K$. When $l=1$, it reduces to a Weighted SBM with Gaussian weights. Utilizing the closed-form expression for the Hellinger distance between Gaussian r.v., the result from \cite{Xu2017OptimalRF} shows that when $p=q$ the misclustering rate is $\exp(-2\log n (1-\exp(-\min_{a \neq b \in [K]}|\mu_{a}-\mu_{b}|^2)))$. Through a first-order Taylor approximation, we obtain $1-\exp(-\min_{a \neq b \in [K]}|\mu_{a}-\mu_{b}|^2) \approx |\mu_{a}-\mu_{b}|^2$, matching our convergence rate up to a constant factor. Notice that under the CSBM \citep{csbm_braun22a}, the SNR is of order $n(p-q)+ \min_{a\neq b} \norm{\mu_{a}-\mu_{b}}^2$. By consequence, the information added by nodes covariate is independent of the sparsity level of the graph. In our setting, it is multiplied by the sparsity level of the graph $np$. Due to the multiplicative effect, integrating edge-side information could have a stronger impact on the SNR than node-side information.
    
\end{remark}

\begin{proof}[Sketch of the proof of Theorem \ref{thm:main}] The result is obtained by using the framework developped by \cite{Gao2019IterativeAF}. The oracle error is controlled by using a conditioning argument, see  Lemma \ref{lem:conc_oracle}. The main challenge, as discussed in Section \ref{sec:noise}, is to control the noise. Since the calculations involved are long, we relegated them to the appendix, Section \ref{app:main}.
\end{proof}

\subsection{Minimax lower-bound}
We are going to show that the convergence rate of Theorem \ref{thm:main} is optimal; up to a constant factor. Assume that the covariates are Gaussian, i.e. $\epsilon_{ij}\sim \calN(0,I_d)$ and consider the following space of parameters \[ \Theta =\lbrace p=q\in [0,1], \mu_{kk'}\in [-1,1]^d, \, \forall k, k' \in [K] \rbrace. \]

\begin{theorem}\label{thm:minimax}
     If $\Delta_{min}\to \infty$, there exists a constant $c'>0$ such that \[ \inf_{\hat{z}} \sup_{\theta \in \Theta} \expec (r(\hat{z},z)) \geq \exp(-c'\Delta_{min}^2). \] If $\Delta_{min}^2=O(1)$, then $\inf_{\hat{z}} \sup_{\theta \in \Theta} \expec (r(\hat{z},z)) \geq c$ for some positive constant $c$.
\end{theorem}

\begin{remark}
    We exclusively examined the extreme scenario where the graph provides no information about the community structure. As the minimax lower-bound is only tight up to a constant factor, extending this result to the case where $p>q$ is straightforward by lower bounding the signal using either its \textcolor{teal}{first} or \textcolor{purple}{second} part, see equation \eqref{eq:delta}.
\end{remark}
\begin{proof}[Sketch of the proof]
    First, we lower bound the minimax risk by the error associated with a two-hypothesis testing problem by using the argument of \cite{gaodcsbm}. Since the optimal test is achieved by the likelihood ratio (Neyman Pearson Lemma), it is sufficient to lower the probability of failure of this optimal test. This can be done by first conditioning on $A$ and using the well-known properties of Gaussian's r.v., and then integrating over $A$, cf. Section \ref{app:minimax} in the appendix for details.
\end{proof}

\subsection{Oracle error}
If we ignore the error terms $F_{ib}^{(t)}$ and $G_{ib}^{(t)}$, a node $i$ is misclassified when $C_i(a,b) <-\Delta^2(a,b)$. This is an unavoidable source of error since it corresponds to the error made by the algorithm after one iteration initialized with the ground-truth partition and with the true model parameters. The error occurring in this way can be quantified by the \textbf{oracle error} defined for all $\delta \in (0,1/2]$
by \[ \xi (\delta) = \sum_{i=1}^n \sum_{b \in [K]\backslash z_i} \Delta^2(z_i,b)\indic_{\lbrace C_i(z_i,b) \leq -(1-\delta)\Delta^2({z_i},b)\rbrace }.\]
\begin{lemma}\label{lem:conc_oracle} Let $\delta \in (0,1/2]$ be a constant, and let us denote for any given $i\in [n]$ and $b\in [K]\setminus z_i$ the event \[\Omega_1(z_i,b) = \left\lbrace C_i(z_i,b) \leq -(1-\delta)\Delta^2({z_i},b) \right\rbrace. \] Under the assumptions of Theorem \ref{thm:main}, there exists a constant $c>0$ such that for all $z_i\neq b$ \[ \prob(\Omega_1(z_i,b))\leq e^{-c\Delta_{min}^2}.\] 
\end{lemma}
\begin{proof}
Assume that $z_i=a$. First, let us decompose $C_i(a,b)+\Delta^2(a,b)$ as $Y_1+Y_2$ where \[ Y_1 =\log(\frac{p}      {q})\left(\sum_{j\in \calC_{a}}A_{ij}-\sum_{j\in \calC_{b}}A_{ij}\right)\] and \[Y_2=\sum_{l\in [K], j\in \calC_{l}}A_{ij}\norm{\mu_{al}-\mu_{bl}}^2+2A_{ij}\langle \epsilon_{ij},\mu_{al}-\mu_{bl}\rangle .\]

Conditionally on $A_{i:}$, we have for all $t\in \R$ 
\begin{align*}
    \expec & (e^{tY_2}|A_{i:})=e^{t\sum_{l\in [K], j\in \calC_{l}}A_{ij}\norm{\mu_{al}-\mu_{bl}}^2}\\ &\times\expec(e^{2t\sum_{l\in [K]} \sum_{j\in \calC_{l}}A_{ij}\langle \epsilon_{ij},\mu_{al}-\mu_{bl}\rangle}|A_{i:})\\
    &\leq e^{(t+2t^2)\sum_{l\in [K]} \sum_{j\in \calC_{l}}A_{ij}\norm{\mu_{al}-\mu_{bl}}^2}. \tag{since $\epsilon_{ij}$ are ind. Sub-Gaussian r.v.}
\end{align*}
 Let us denote \[\Delta^2_A=\Delta^2_A(a,b)=\sum_{l\in [K], j\in \calC_{l}}A_{ij}\norm{\mu_{al}-\mu_{bl}}^2.\]  We have shown that for all $t$ \[
 \expec(e^{t(Y_1+Y_2)})\leq \expec(e^{tY_1+(t+2t^2)\Delta^2_A}).
 \] 
 We can rewrite $tY_1+(t+2t^2)\Delta^2_A$  as a weighted sum of independent Bernoulli trials $\sum_{j} w_{j}(t)A_{ij} $ where $w_{j}(t)=t\log(p/q)+(t+2t^2)\norm{\mu_{aa}-\mu_{ba}}^2$ for $j\in \calC_a$, $w_{j}(t)=-t\log(p/q)+(t+2t^2)\norm{\mu_{ab}-\mu_{bb}}^2$ for $j \in \calC_b$ and  $w_{j}(t)=(t+2t^2)\norm{\mu_{al}-\mu_{bl}}^2$ when $j \in \calC_l$ for $l\neq a, b$.
 Hence, we obtain 
 \begin{align*}
     \log &\expec(e^{tY_1+(t+2t^2)\Delta^2_A})\leq \sum_{j}p_{ij}(e^{w_{j}(t)}-1)\\
     &\leq \sum_j p_{ij}w_j(t)+0.5e^{\sup_j |w_j(t)|} \sum_j p_{ij}w_{j}(t)^2 \tag{by Taylor-Lagrange formula}.
\end{align*}
Since $p/q=O(1)$ and $\norm{\mu_{al}-\mu_{bl}}^2\leq 4$ for all $l\in [K]$, one can choose $t^*<0$ close to $0$ such that \[ e^{\sup_j |w_j(t^*)|} \sum_j p_{ij}w_{j}(t^*)^2 \leq |\sum_j p_{ij}w_j(t^*)|\] and \[ \sum_j p_{ij}w_j(t^*) \leq -c'\Delta^2(a,b)\] for some positive constant $c'$.  
By consequence
\begin{align*}
    \prob(\Omega_1(z_i,b))&= \prob(Y_1+Y_2\leq \delta \Delta^2(a,b))\\
    &\leq \expec(e^{t^*(Y_1+Y_2)})e^{-t^*\delta\Delta^2(a,b)}\\
    & \leq e^{-0.5c'\Delta^2(a,b)-t^*\delta\Delta^2(a,b)}\\
    &\leq e^{-0.25c'\Delta^2(a,b)}
\end{align*}
for all $\delta >0$ smaller than $\min\lbrace c'|4t^*|^{-1},1/2\rbrace=1/2$.
\end{proof}

\begin{corollary}\label{cor:conc_oracle}Under the assumptions of Theorem \ref{thm:main}, with probability at least $1-e^{-\Delta_{min}^2}$ we have for some constant $c>0$\[\ \xi(\delta)\leq ne^{-c\Delta_{min}^2}. \]
\end{corollary}
\begin{proof}
    By Lemma \ref{lem:conc_oracle}, we have \[ \expec(\xi(\delta)) \leq ne^{-(c-o(1))\Delta_{min}^2} \] since $K$ is constant. Hence, by Markov inequality
    \begin{align*}
        \prob(\xi(\delta)\geq e^{\Delta_{min}} \expec(\xi(\delta))) \leq e^{-\Delta_{min}}.
    \end{align*}
    Since $e^{\Delta_{min}} \expec(\xi(\delta)))\leq ne^{-(c+o(1))\Delta^2_{min}}$ we obtain the result of the lemma.
\end{proof}

\subsection{Control of the noise}\label{sec:noise}
To apply Theorem 3.1 in \cite{Gao2019IterativeAF} to show that the error contracts at each step until reaching the oracle error, one needs to prove that the noise terms $F_i^{(t)}$ and $G_i^{(t)}$ satisfy the following conditions.
Let $\tau = \epsilon n\Delta_{min}^2/K^2$ where $\epsilon > 0$ and let $\delta \in (0,1/2)$ be a constant. Let us define the weighted Hamming loss  \[ l (z,z')= \sum_{i=1}^n\Delta^2(z_i,z'_i)\indic_{\lbrace z_i\neq z_i' \rbrace }.\] 
\begin{condition}[F-error type]\label{cond:f}
Assume that \[ \max_{\lbrace z^{(t)}: l(z,z^{(t)})\leq \tau \rbrace}\sum_{i=1}^n\max_{b\in [K]\backslash z_i} \frac{(F_{ib}^{(t)})^2}{\Delta^2(z_i,b)l(z,z^{(t)})} \leq \frac{\delta^2}{256}\] for all $t \geq 0$ holds with probability at least $1-n^{-\Omega(1)}$.
\end{condition}

\begin{condition}[G-error type]\label{cond:h}
Assume that \[ \max_{i\in [n]}\max_{b \in [K]\setminus z_i} \frac{|G_{ib}^{(t)}|}{\Delta^2 (z_i,b) } \leq \frac{\delta}{4} \] holds uniformly on the event $\lbrace z^{(t)}: l(z,z^{(t)})\leq \tau \rbrace$ for all $t\geq 0$ with probability at least $1-n^{-\Omega(1)}$ .
\end{condition}
Condition \ref{cond:f} necessitates a uniform control of the noise induced by the $F^{(t)}$ error term in an $\ell_2$ norm sense. Additionally, Condition \ref{cond:h} requires a uniform control of the $l_\infty$ norm of the $G^{(t)}$ error term. Typically, the $F^{(t)}$-error term depends on the estimation error of the partition $\norm{Z^{(t)}-Z}$ while the $G^{(t)}$-error term depends on the parameter estimation error $\norm{\Pi^{(t)}-\Pi}$.

The main technical challenge to prove the consistency or \sir\, is to show that the previous conditions hold. In particular, one needs to control the estimation error of the model parameters uniformly. For the SBM part, it can be done as in \cite{csbm_braun22a}, but bounding uniformly the error $\norm{\mu_{kk'}-\mu_{kk'}^{(t)}}$ requires a new approach: contrary to the CSBM setting, the edge centroids are estimated on random samples that depend on $A$ and the current estimate of the partition $\calC^{(t)}$. This is the object of the following lemma.
\begin{lemma}\label{lem:fond_ineq_bis}
    Under the assumption of Theorem \ref{thm:main} we have with probability at least $1-n^{-\Omega(1)}$, for all $z^{(t)}$ such that $\lt \leq \tau $
        \[\max_{b, l\in [K]} \norm{\mu_{bl}-\mu_{bl}^{(t)}} \lesssim K^{1.5}\left(\sqrt{\frac{\lt}{n\Delta_{min}^2} }\vee \frac{\sqrt{\log K}}{np_{max}}\right).\] 
\end{lemma}

\begin{proof}[Sketch of the proof.]

To obtain an uniform bound over $z^{(t)}$, we first need to control uniformly over all set $S$ the quantities $\norm{\sum_{i\in S} \epsilon_i}$ where $\epsilon_i$ are i.i.d. sub-Gaussian r.v. This can be done by using Lemma A.1 in \cite{Lu2016StatisticalAC}. Secondly, we need to control uniformly over $T_1, T_1\subset [n]$ the sums $\sum_{i\in T_1,j\in T_2}A_{ij}$. This can be done by showing that $A$ satisfies the discrepancy property, e.g. \cite{rinaldo2015}. The details can be found in the appendix (Section \ref{app:f_err}, Lemma \ref{lem:fond_ineq}).
\end{proof}

 We will also need control of the following term to establish Condition \ref{cond:f}.
\begin{lemma}\label{lem:conc_quadr}
       Under the assumption of Theorem \ref{thm:main} we have with probability at least $1-n^{-\Omega(1)} $ \[ \sum_{i\in \calC_{k}}\sum_{j \in \calC_{k'}, j'\in \calC_{k''}}  E_{ij} E_{ij'} \lesssim (n/K)^2p_{max},\] for all  $ k, k',k''\in [K].$
\end{lemma}
\begin{proof}Let's fix $k, k', k''\in [K]$. To simplify the exposition, we will assume that each class is of size $n/K$. We want to bound \[ S=\sum_{j, j'} \sum_{i \in \calC_k}E_{ij}E_{ij'}=\sum_{j, j'}\langle E_{:j}, E_{:j'}\rangle_{\calC_{k}}, \] where the scalar product is restricted to the entries of $E_{:j}$ in $\calC_{k}$. In the following, we will drop the subscript when clear from the context. 

\textbf{Case where $k, k'$, and $k''$ are all distinct.} $S$ is a sum of $(n/K)^3$ bounded and centered independent r.v. with variance of order $p_{max}^2$  that can be handled with standard concentration inequalities.

\textbf{Case where $k\neq k'=k''$.} We have $S=\sum_{j \neq j'}  \langle E_{j:}, E_{:j'}\rangle  + \norm{E}_F^2$. It is easy to show that w.h.p. $\norm{E}_F^2\lesssim (n/K)^2p_{max}$. So we can focus on the first summand. Observe that $\sum_{j\neq  j'} \expec(\langle E_{:j}, E_{:j'}\rangle) = 0$. To remove the dependencies in the sum, we will use a decoupling argument similar to the one used in \cite{HS} to prove Hanson-Wright inequality. This strategy has also been used by \cite{braun2023strong} in the setting of bipartite graphs. 

 Let $(\delta_j)_{j\in [n]}$ be independent Bernoulli r.v.~with parameter $1/2$ and let us define the set of indices
$\Lambda_\delta = \lbrace j\in \calC_{k'}: \delta_j=1\rbrace $ and the random variable \[ S_\delta =\sum_{j,j'\in \calC_{k'}}\delta_j(1-\delta_{j'})\langle E_{j:},E_{j':} \rangle= \sum_{j\in \Lambda_\delta} \langle E_{j:}, \sum_{j'\in \Lambda_\delta^c}E_{j':}\rangle.\]
Let us denote by $\expec_{\Lambda^c}(.)$ the conditional expectation on $\delta$ and $(E_{j':})_{j'\in \Lambda_\delta^c}$. For all $t>0$, we have 
\begin{align*}
    \log & \, \expec_{\Lambda^c} \left( e^{tS_\delta} \right) =\sum_{i, j\in \Lambda_\delta}\log(\expec_{\Lambda_\delta^c}( e^{E_{ij}t\sum_{j'\in \Lambda_\delta^c}E_{ij'}})\\
    & = \sum_{i, j\in \Lambda_\delta}\left(\log(\expec e^{A_{ij}t\sum_{j'\in \Lambda_\delta^c}E_{ij'}})-p_{ij}t\sum_{j'\in \Lambda_\delta^c}E_{ij'}\right)\\
    &\leq \sum_{i, j\in \Lambda_\delta}\left( p_{ij}(e^{t\sum_{j'\in \Lambda_\delta^c}E_{ij'}}-1) - p_{ij}t\sum_{j'\in \Lambda_\delta^c}E_{ij'}\right) \tag{$\log(1+x)\geq x$, for all $x>-1$}\\
    &\leq e^{t\max_i\sum_{j'\in \Lambda_\delta^c}E_{ij'}}0.5t^2p_{max}\sum_{i\in \calC_k, j\in \calC_{k'}} (\sum_{j'\in \Lambda_\delta^c}E_{ij'})^2 \tag{by Taylor-Lagrange formula}.
\end{align*} 
Let $C_1>1$ be an appropriately large constant and let us denote the events \tiny\begin{align*}
     \calE(\Lambda_\delta^c) &=\lbrace \sum_{i\in \calC_k,j\in\calC_{k'}} (\sum_{j'\in \Lambda_\delta^c} E_{ij'})^2 \leq C_1 (n/K)^3p_{max} \rbrace\\& \cap \lbrace \max_i\sum_{j'\in \Lambda_\delta^c}E_{ij'} \leq C_1np_{max}/K \rbrace
\end{align*} \normalsize and\[\calE =\lbrace \max_{\Lambda_\delta}\sum_{i\in \calC_k,j\in\calC_{k'} }(\sum_{j'\in \Lambda_\delta^c} E_{ij'}) ^2\leq C_1(n/K)^3p_{max}\rbrace \cap \calD\] where $\calD=\lbrace \max_i \sum_{j}A_{ij}\leq C_1np_{max}/K\rbrace$.
  By Bernstein inequality (cf. appendix Section \ref{app:lem3}) $\calE$ occurs with probability at least $1-n^{-5}$ for $C_1$ large enough. By choosing $t=(C_1np_{max}/K)^{-1}$, we obtain   for $C_2>1$ large enough \begin{align*}
    \prob &\left(S_\delta \geq C_2(n/K)^2p_{max} \cap \calE\right)\\ &\leq  \expec\left(\indic_{\calE(\Lambda_\delta^c)}\prob (S_\delta \geq C_2 (n/K)^2p_{max}|\Lambda^c)\right)\\
    &\leq  \expec\left(\indic_{\calE(\Lambda_\delta^c)}\prob (e^{tS_\delta} \geq e^{C_2t(n/K)^2p_{max}}|\Lambda^c)\right)\\
    &\leq e^{-tC_2(n/K)^2p_{max}}\expec \left(\indic_{\calE(\Lambda_\delta^c)} \expec_{\Lambda^c} \left( e^{tS_\delta} \right)\right)\\
    &\lesssim e^{-C_2C_1^{-1}(n/K)}e^{2n/(C_1K)} \leq e^{-5n/K}.
\end{align*}
By an union bound argument, \[ \prob(\underbrace{\exists \delta, S_\delta \gtrsim (n/K)^2p_{max} \cap \calE}_{\calE_1}) \leq 2^{n/K}e^{-5n/K}\leq e^{-4n/K}. \] Hence, $\prob(\exists \delta, S_\delta \gtrsim (n/K)^2p_{max})\leq e^{-4n/K} + n^{-5}$.
%
Since\[ \sum_{j\neq j'} \langle E_{:j}, E_{:j'}\rangle_{\calC_{k}} = 4\expec_\delta (S_\delta),\] we obtain that $S\lesssim (n/K)^2p_{max}$ with probability at least $1-n^{-5}$.

\textbf{Case $k=k'=k''$.}  See the appendix, Section \ref{app:lem3}. 
\end{proof}

%% file: num_xp.tex
\section{Numerical experiments}\label{sec:xp}
In this section, we evaluate our proposed algorithms, \sir\, and \ir\, (with $T=3$), on synthetic data and the email EU core dataset \citep{Leskovec-2007-evolution} with synthetic edge covariates\footnote{The experiments were conducted using R on a CPU Intel Core i7-1255U. The code implementation can be found on \url{https://github.com/glmbraun/VECSBM/}.}. 

We compare our methods with \olmf\, \citep{paul2020}, a general matrix factorization approach applicable beyond the multi-layer graph setting, and the vanilla spectral method \spec\, which doesn't incorporate edge side information. We assess the accuracy of clustering using the Normalized Mutual Information (NMI) criterion, where an NMI of zero indicates no significant correlation between the clusters, and an NMI of one signifies a perfect match.

\subsection{Network with indistinguishable communities (Scenario 1)}
We consider a VEC-SBM with $K=3$, $n=600$, and such that the graph is generated by an SBM with parameters $p=3.5\log n$ and $q=\log n$ where the communities $1$ and $2$ are indistinguishable. The covariates are such that $\mu_{11}=(1,1,1)$ and $\mu_{22}=-\mu_{11}$ and all the other centroids are zero. Thus, the covariates only separate communities $1$ and $2$.  As shown in Figure \ref{fig:boxplot}, \ir\, and \sir\, outperform \olmf\, and effectively combine both sources of information to recover the clusters. However, we observed that \ir\, is more sensitive to initialization than \sir. This is why we initialized it with \sir. 
\begin{figure}
    \centering
    \includegraphics[scale=0.45]{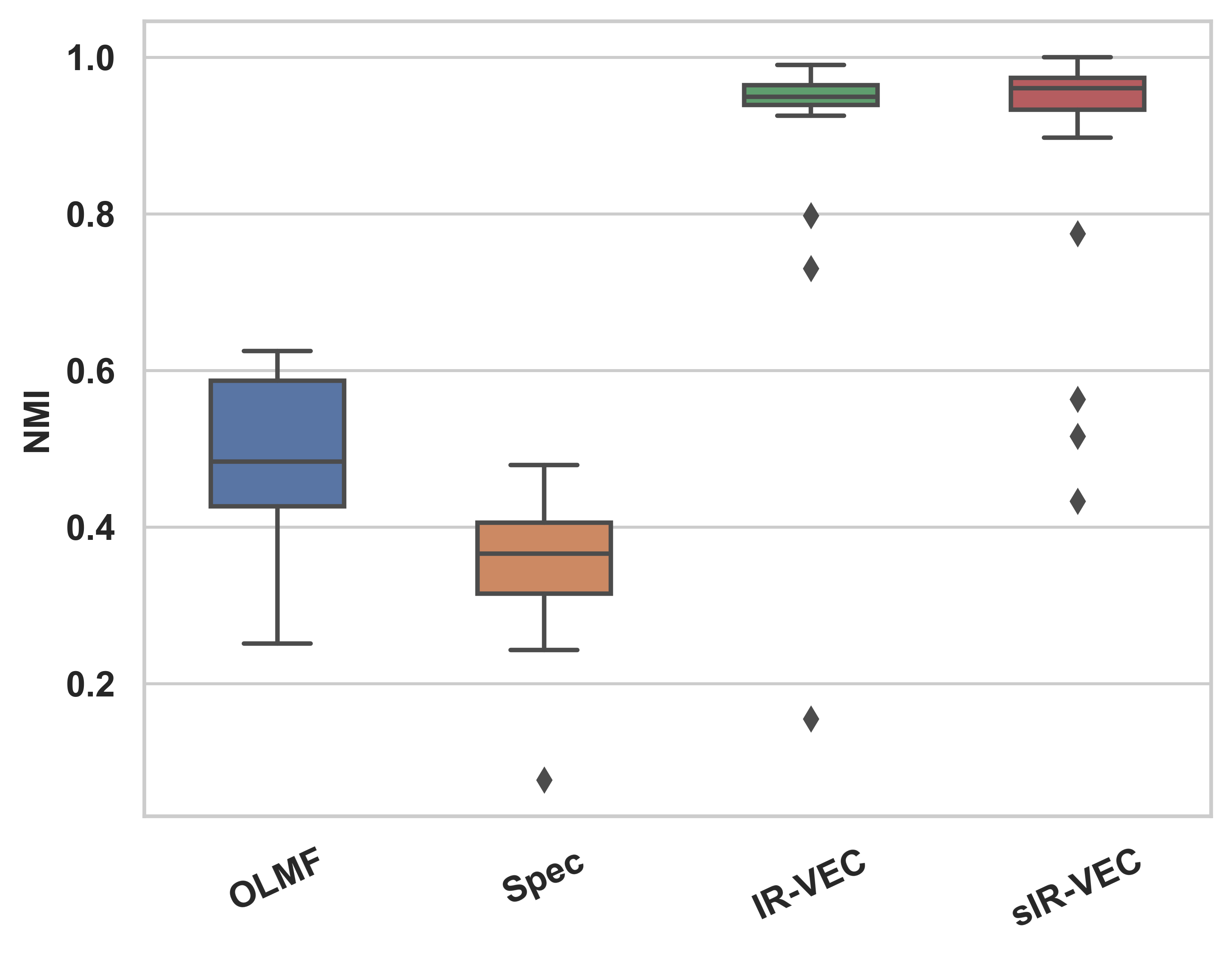}
    \caption{{\small Average performance over $20$ runs under Scenario 1.}}
    \label{fig:boxplot}
\end{figure}

\subsection{Non-isotropic covariance (Scenario 2)}
In this scenario, we sample a VEC-SBM with the same parameters as the previous experiment, except for the edge covariates. Here, $\mu_{kk'}$ is generated uniformly over $[-1,1]$, and $\Sigma_{kk'}$ are positive definite matrices randomly generated using the \texttt{clusterGeneration} package, with the maximal singular value set to $1$. As shown in Figure \ref{fig:res3}, the performance of \olmf\, significantly decreases under this scenario, while \ir\, and surprisingly \sir\, recover accurately the clusters.

\begin{figure}
    \centering
    \includegraphics[scale=0.45]{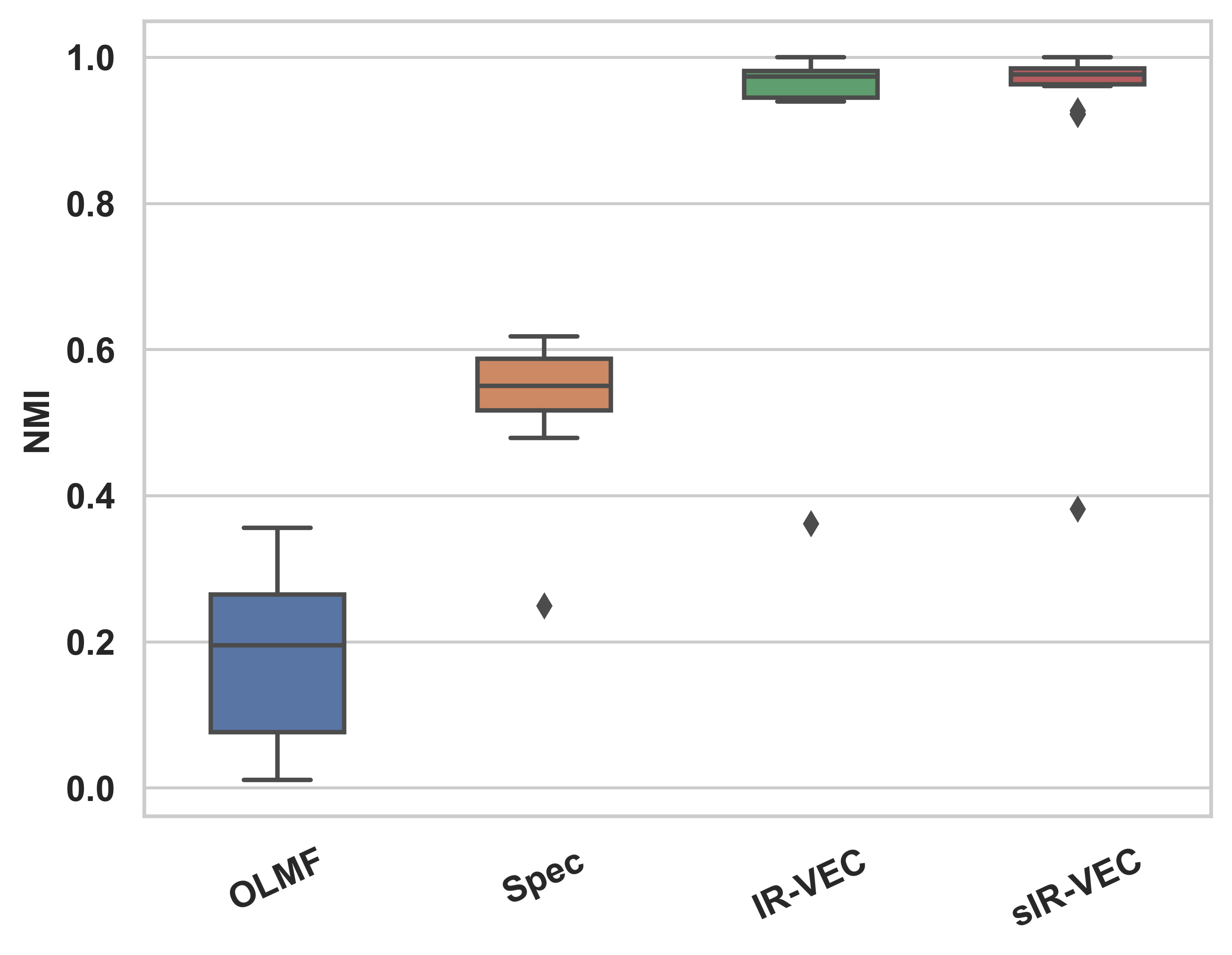}
    \caption{{\small Average performance over $20$ runs under Scenario 2.}}
    \label{fig:res3}
\end{figure}

\subsection{Influence of the number of communities (Scenario 3)}
We evaluate the performance of our method as the number of communities increases. We fix the parameters: $n=1000$, $p=8\log n/n$, $q=p/2$, and generate isotropic edge covariates with centroids sampled uniformly over $[-2,2]$ for  $K \in \lbrace 2, 4, 6, 8, 10\rbrace$. As shown in Figure \ref{fig:k}, while the spectral method's performance decreases with increasing $K$, \sir\, is less sensitive. This is because the edge distribution is dissymmetric, allowing the SNR to remain higher when $K$ increases. Additionally, we observe that \sir\, performs well when initialized with an almost uninformative $z^{(0)}$ provided by \spec.

\begin{figure}
    \centering
    \includegraphics[scale=0.45]{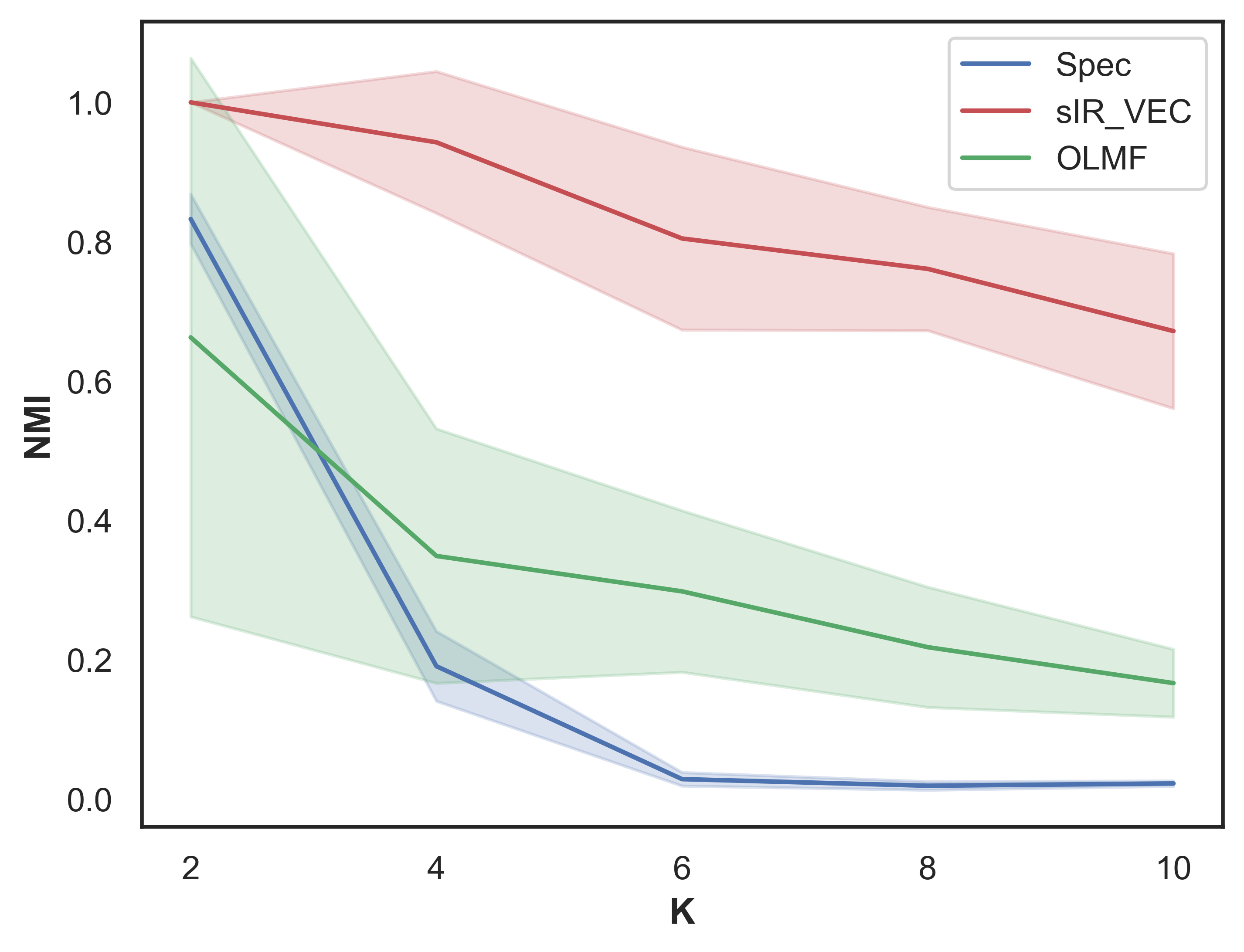}
    \caption{{\small Average performance over $20$ runs with varying $K$ (Scenario 3).}}
    \label{fig:k}
\end{figure}

\subsection{Email EU core dataset}

The dataset \citep{Leskovec-2007-evolution} comprises email communications between members of different European research institutions. We restrict the dataset to six institutions with at least $50$ members and consider the institution as the ground-truth partition. Isolated vertices are removed, and for each edge, we simulate a textual distribution across 6 topics depending on the communities of its endpoints. The proportion of topics for each $k,k'\in [K]$ is generated uniformly over $[0,1]$. We obtained an NMI of $0.49$ with the spectral method, while \ir\, is more accurate and provides a clustering with an NMI of $0.81$ after $15$ iterations. The number of iterations required is higher than when the graph is generated from an SBM, but \ir\, appears robust to variations in graph topology.

%% file: appendix.tex
\onecolumn
\begin{center}
\Large \textbf{Supplementary Material}
\end{center}

The proof of Theorem \ref{thm:main} is presented in Section \ref{app:main}. More precisely, Section \ref{app:err_dec} gives the exact expression of the error decomposition discussed in Section \ref{sec:err_dec}. Section \ref{app:f_err} shows that the F-error term satisfies Condition \ref{cond:f} and Section \ref{app:g_err} shows that the G-error term satisfies Condition \ref{cond:h}. The proof of the remaining case of Lemma \ref{lem:conc_quadr} is given in Section \ref{app:lem3}.

\section{Proof of Theorem \ref{thm:main}}\label{app:main}
We will use the following notations.  Let us define the Hamming loss $h$ as \[ h(z_i,z'_i)=\sum_i\indic_{z_i\neq z_i'}\] and recall that $l$ was defined as the weighted Hamming loss \[ l (z,z')= \sum_i\Delta^2(z_i,z'_i)\indic_{z_i\neq z_i'}.\] We will denote by $h^{(t)}$  the corresponding function applied to $z$ and $z^{(t)}$.
\subsection{Error decomposition}\label{app:err_dec}
Let us define for all $k, k'\in [K]$ the oracle estimators \[ \tilde{p}= \sum_{k\in [K]}\frac{\sum_{i, j \in \calC_k}A_{ij}} {K|\calC_k|^2}, \tilde{q}= \sum_{k\neq k'\in [K]}\frac{\sum_{i\in \calC_k, j\in \calC_{k'}}A_{ij}} {K(K-1)|\calC_k||\calC_k'|}, \tilde{\mu}_{kk'}=\frac{\sum_{i\in \calC_k, j\in \calC_{k'}}A_{ij}G_{ij}} {|\calC_k||\calC_k'|}.\]
The node $i$ such that $z(i)=a$ is incorrectly classified at step $t$ iff \[ a \neq \arg\max_k MAP_i(\calC^{(t)}, \Pi^{(t)}, \mu^ {(t)},I_d).\] It implies that there exists a $b\neq a \in [K]$ such that 
\[ \sum_{l\in [K]} \sum_{j\in \calC^{(t)}_{l,-i}} A_{ij} \left(\log (\Pi^{(t)}_{al})-\log (\Pi^{(t)}_{bl}) - \norm{G_{ij}-\mu^{(t)}_{al}}^2+\norm{G_{ij}-\mu^{(t)}_{bl}}^2\right)<0. \]
This last condition can be further decomposed as 
\[ \underbrace{\log(\frac{p}{q})\left(\sum_{j\in \calC_{a}}E_{ij}-\sum_{j\in \calC_{b}}E_{ij}\right)+\sum_{l\in [K]} \sum_{j\in \calC_{l,-i}}\left(E_{ij}\norm{\mu_{al}-\mu_{bl}}^2+2A_{ij}\langle \epsilon_{ij},\mu_{al}-\mu_{bl}\rangle\right)}_{C_i(a,b)}< -\Delta^2(a,b)+F_i^{(t)}+G_i^{(t)}\]
where \begin{align*}
    \Delta^2(a,b)&= \log (p/q)(|\calC_a|p\gb{-}|\calC_b|q)+\sum_{l\in [K]}|\calC_l|\Pi_{al}\norm{\mu_{al}-\mu_{bl}}^2,\\
    F_i^{(t)} &= \langle E_{i:}(Z^{(t)}-Z), \log \Pi_{a:}- \log \Pi_{b:}\rangle+ \langle E_{i:}Z^{(t)}, \log \Pi_{a:}^{(t)}- \log \tilde{\Pi}_{a:}-\log\Pi_{b:}^{(t)}+ \log \tilde{\Pi}_{b:}\rangle\\
    &+\sum_{l\in [K]}\sum_{j\in \calC_l}\left(2A_{ij}\langle \epsilon_{ij}, \mu_{al}^{(t)}-\tilde{\mu}_{al}+\tilde{\mu}_{bl}-\mu_{bl}^{(t)}\rangle + E_{ij}\norm{\tilde{\mu}_{bl}-\mu_{bl}^{(t)}}^2-2E_{ij}\langle \mu_{al}-\mu_{bl}, \tilde{\mu}_{bl}-\mu_{bl}^{(t)}\rangle \right),\\
    G_i^{(t)}&= \langle P_{i:}(Z^{(t)}-Z), \log \Pi_{a:}- \log \Pi_{b:}\rangle+ \langle P_{i:}Z^{(t)}, \log \Pi_{a:}^{(t)}- \log \Pi_{a:}+\log\Pi_{b:}^{(t)}+ \log \Pi_{b:}\rangle \\
    &+\sum_{l\in [K]}\sum_{j\in \calC_l}P_{ij}\left(\norm{\mu_{bl}-\mu_{bl}^{(t)}}^2-2P_{ij}\langle \mu_{al}-\mu_{bl}, \mu_{bl}-\mu_{bl}^{(t)}\rangle \right)\\
    &+\langle E_{i:}Z^{(t)}, \log \tilde{\Pi}_{a:}-\log \Pi_{a:}+\log \Pi_{b:}-\log \tilde{\Pi}_{b:}\rangle -2\sum_{l\in [K]}\sum_{j\in \calC_l}E_{ij}\langle \mu_{al}-\mu_{bl}, \mu_{bl}-\tilde{\mu}_{bl}\rangle \\
    &+ \sum_{l\in [K]}\sum_{j\in \calC_l}E_{ij}\left(\norm{\mu_{bl}-\mu_{bl}^{(t)}}^2-\norm{\tilde{\mu}_{bl}-\mu_{bl}^{(t)}}^2\right)+2A_{ij}\langle \epsilon_{ij}, \tilde{\mu}_{al}-\mu_{al}-\tilde{\mu}_{bl}+\mu_{bl}\rangle.
\end{align*}

The term $\Delta^2(a,b)$ is deterministic and corresponds to the signal. Under Assumptions \ref{ass:balanced_part} and \ref{ass:sym}, it is easy to see that \[ \Delta^2(a,b)\asymp \frac{n(p-q)}{K}+\frac{np}{K}\sum_l\norm{\mu_{al}-\mu_{bl}}^2.\] So if the difference between the centroids  $\sum_l\norm{\mu_{al}-\mu_{bl}}^2= \Omega (1)$, edges covariate have a multiplicative effect on the signal. The error terms $(F_i{(t)})_{i\in [n]}$ depend linearly on $\epsilon_{i:}$ and $E_{i:}$ and these errors can be controlled in average. On the other hand, the error terms $(G_i{(t)})_{i\in [n]}$ need to be controlled uniformly. 

\subsubsection{F-error term}\label{app:f_err}
We need an upper-bound of \[ F=\max_{\lbrace z^{(t)}: l(z,z^{(t)})\leq \tau \rbrace}\sum_{i=1}^n\max_{b\in [K]\backslash z_i} \frac{(F_{ib}^{(t)})^2}{\Delta^2(z_i,b)l(z,z^{(t)})}.\]
Toward this end, notice that $(F_{ib}^{(t)})^2 \lesssim F_1+F_2+F_3+F_4$ where
\begin{align*}
    F_1 &=\sum_{i=1}^n\max_{b\in [K]\backslash z_i} \langle E_{i:}(Z^{(t)}-Z), \log \Pi_{a:}- \log \Pi_{b:}\rangle ^2\\
    F_2 &= \sum_{i=1}^n\max_{b\in [K]\backslash z_i}\langle E_{i:}Z^{(t)}, \log \Pi_{a:}^{(t)}- \log \tilde{\Pi}_{a:}-\log\Pi_{b:}^{(t)}+ \log \tilde{\Pi}_{b:}\rangle ^2\\
    F_3 &= \sum_{i=1}^n\max_{b\in [K]\backslash z_i}\left(\sum_{l\in [K]}\sum_{j\in \calC_l}A_{ij}\langle \epsilon_{ij}, \mu_{al}^{(t)}-\tilde{\mu}_{al}+\tilde{\mu}_{bl}-\mu_{bl}^{(t)}\rangle\right)^2 \\
    F_4 &= \sum_{i=1}^n\max_{b\in [K]\backslash z_i}\left(\sum_{l\in [K]}\sum_{j\in \calC_l}E_{ij}\left(\norm{\tilde{\mu}_{bl}-\mu_{bl}^{(t)}}^2-2\langle \mu_{al}-\mu_{bl}, \tilde{\mu}_{bl}-\mu_{bl}^{(t)}\rangle \right) \right)^2 . 
\end{align*}

First, let us establish some useful inequalities that will be used repeatedly.
\begin{lemma}\label{lem:fond_ineq}
    Under the assumption of Theorem \ref{thm:main} we have with probability at least $1-n^{-\Omega(1)}$, for all $z^{(t)}$ such that $K\lt \leq n\Delta_{min}^2\epsilon $
    \begin{enumerate}
        \item $\max_{k\in [K]}|n_k^{(t)}-n_k| \leq \frac{ l(z^{(t)},z)}{\Delta_{min}^2}$
        \item $ \max_{k\in [K]}|| Z^{(t)}_{:k}-Z_{:k}|| \leq \norm{Z^{(t)}-Z } \lesssim \frac{K^{0.5}}{n^{0.5}\Delta_{min}^2} l(z^{(t)},z),$
        \item $ |\log(\frac{p^{(t)}}{q^{(t)}})-\log(\frac{\tilde{p}}{\tilde{q}})|\lesssim K\frac{\lt}{n\Delta^2_{min}}$
        \item $ \max_{b, l\in [K]} \norm{\tilde{\mu}_{bl}-\mu_{bl}^{(t)}} \lesssim K^{1.5}\left(\sqrt{\frac{\lt}{n\Delta^2_{min}} } \vee \frac{\sqrt{\log K}}{np_{max}}\right)$ ,
        \item $ \max_{b, l\in [K]} \norm{\tilde{\mu}_{bl}-\mu_{bl}} \lesssim \frac{K}{\sqrt{n^2p_{max}}}$.
    \end{enumerate}
\end{lemma}

\begin{proof}
    The first three items are direct consequences of Lemma 13, 14, and 15 in \cite{csbm_braun22a}. However, the fourth point doesn't derive immediately from Lemma 4.1 in \cite{Gao2019IterativeAF} since in our setting the estimate of $\mu_{bl}$ depends on $A$. We will need several properties of the adjacency matrix $A$ that holds w.h.p. We recall them below.
    \paragraph{Fact 1.}  The adjacency matrix $A$ satisfies the following version of the discrepancy property (see e.g. \cite{rinaldo2015}) \[ \forall T_1, T_2 \subset [n] \text{ such that }|T_1|\leq |T_2|,\,  T(A)\leq \kappa(|T_1|, |T_2|)p_{max}|T_1||T_2| \]  where $T(A)=\sum_{(i,j)\in T_1\times T_2}A_{ij}$ and $\kappa (|T_1|, |T_2|) = \max (\frac{c^*\log (en/|T_2|)}{p_{max}|T_1|}, C^*)$ for constants $c^*, C^*>0$.The argument is based on Chernoff's bound and a union bound, see the supplementary material of \cite{rinaldo2015} for more details.
    \paragraph{Fact 2.} Let us denote by $I_A(a,b)=\sum_{i\in \calC_a, j \in \calC_b}A_{ij}$ the number of observed edges between communities $a$ and $b\in [K]$. By using Chernoff's bound and a union bound over $(a,b)\in [K]^2$ we can show that w.h.p for all $a,b \in [K]$ \[ I(a,b) \asymp p_{max}\frac{n^2}{K^2}.\]

    We will assume from now on that $A$ satisfies the previous properties. Hence, conditionally on $A$, a straightforward adaptation of Lemma A.1 in \cite{Lu2016StatisticalAC} implies that w.h.p. \begin{equation}\label{eq:noise}
        \forall T=T_1\times T_2\subset [n]\times [n], \, \norm{\sum_{(i,j)\in T}A_{ij}\epsilon_{ij} }\lesssim \sqrt{|T(A)|n}.
    \end{equation} The price to pay to have a uniform bound is an additional $\sqrt{n}$ factor.
    
    Consider the set $T= (\calC_a\times \calC_b)\Delta (\calC_a^{(t)}\times \calC_b^{(t)})$ where $\Delta$ denotes the symmetric difference operator. It can be further decomposed as $\cup_{c=1}^3T^{(c)}$ where $T^{(1)}= \calC_a\setminus C_a^{(t)}\times  \calC_b$, $T^{(2)}= (\calC_a\cap \calC_a^{(t)})\times  (\calC_b^{(t)}\Delta  \calC_b)$, and $T^{(3)}= \calC_a^{(t)}\setminus \calC_a\times  \calC_b^{(t)}$. Note that by definition of $h^{(t)}$ we have \begin{align*}
    \max_b|\calC_b^{(t)}\Delta  \calC_b|&\leq h^{(t)}\\
    \max_a |\calC_a\setminus C_a^{(t)}|&\leq h^{(t)}\\
     \max_a |\calC_a^{(t)}\setminus C_a|&\leq h^{(t)}.
    \end{align*}
    Notice that by definition of $h^{(0)}$, we have $|\calC_a\cap C_a^{(t)}|\asymp n/K$ and $|\calC_b^{(t)}|\asymp n/K$.
    By Fact 1 and 2, we have for all $c=1\ldots 3$ \begin{equation}\label{eq:card}
        |T^{(c)}(A)| \lesssim p_{max}\frac{nh^{(t)}}{K} \vee 
        \log(K)\frac{n}{K}.
    \end{equation} 
    Let us denote by $I_A^{(t)}(a,b)=\sum_{i\in \calC_a^{(t)}, j \in \calC_b^{(t)}}A_{ij}$ the number of edges between nodes in the estimated communities $a$ and $b$. Since $h^{(t)}\leq h^{(0)}\leq \epsilon n/K^2 $ we have, \[ \abs{I_A^{(t)}(a,b)- I_A(a,b)} \lesssim \epsilon I_A(a,b)/K .\] This implies that $I_A^{(t)}(a,b)\asymp p_{max}n^2/K^2$.
    Now, we can decompose \begin{align*}
        \tilde{\mu}_{bl}-\mu_{bl}^{(t)}& = \frac{\sum_{i_\in \calC_a, j_\in \calC_b}A_{ij}G_{ij}-\sum_{i_\in \calC_a^{(t)}, j_\in \calC_b^{(t)}}A_{ij}G_{ij}}{I_A(a,b)}+\left(\frac{1}{{I_A(a,b)}}-\frac{1}{{I^{(t)}_A(a,b)}}\right)\sum_{i_\in \calC_a^{(t)}, j_\in \calC_b^{(t)}}A_{ij}G_{ij}\\
        &= \underbrace{\frac{\sum_{i_\in \calC_a, j_\in \calC_b}A_{ij}\epsilon_{ij}-\sum_{i_\in \calC_a^{(t)}, j_\in \calC_b^{(t)}}A_{ij}\epsilon_{ij}}{I_A(a,b)}}_{L_1}+\underbrace{\left(\frac{1}{{I_A(a,b)}}-\frac{1}{{I^{(t)}_A(a,b)}}\right)\sum_{i_\in \calC_a^{(t)}, j_\in \calC_b^{(t)}}A_{ij}\epsilon_{ij}}_{L_2}\\
        &+\underbrace{\frac{\sum_{i_\in \calC_a, j_\in \calC_b}A_{ij}\mu_{ab}-\sum_{i_\in \calC_a^{(t)}, j_\in \calC_b^{(t)}}A_{ij}\mu_{z(i)z(j)}}{I_A(a,b)}}_{L_3}\\
        &+\underbrace{\left(\frac{1}{{I_A(a,b)}}-\frac{1}{{I^{(t)}_A(a,b)}}\right)\sum_{i_\in \calC_a^{(t)}, j_\in \calC_b^{(t)}}A_{ij}\mu_{z(i)z(j)}}_{L_4}.
    \end{align*}

\paragraph{Control of $L_1$.} We have \begin{align*}
     \norm{L_1} &= \norm{\frac{\sum_{c=1}^3\sum_{(i,j)\in T^{(c)}}A_{ij}\epsilon_{ij}}{I_A(a,b)}}\\
     &\lesssim K^{1.5}\left( \sqrt{\frac{h^{(t)}}{n}} \vee \frac{\sqrt{\log K}}{np_{max}}\right) \leq K^{1.5}\left( \sqrt{\frac{\lt}{n\Delta^2_{min}}} \vee \frac{\sqrt{\log K}}{np_{max}}\right) \tag{by equations \eqref{eq:card} and \eqref{eq:noise}}.
\end{align*}
\paragraph{Control of $L_3$.} By a similar argument we obtain \begin{align*}
    \norm{L_3} &\lesssim \frac{\sum_c\sum_{(i,j)\in T^{(c)}}A_{ij}}{I_A(a,b)}\\
        &\lesssim K\frac{h^{(t)}}{n}\vee K\frac{\log K}{np_{max}}.
\end{align*}

\paragraph{Control of $L_2$.}
By equation \eqref{eq:noise} 
and the discrepancy property, we obtain \[ \norm{\sum_{i_\in \calC_a^{(t)}, j_\in \calC_b^{(t)}}A_{ij}\epsilon_{ij}}\lesssim \sqrt{n}\sqrt{I(a,b)}.\]
Furthermore, we have \[ \abs{\frac{1}{{I_A(a,b)}}-\frac{1}{{I^{(t)}_A(a,b)}}}\lesssim \frac{|I_A(a,b)-I^{t}_A(a,b)|}{I(a,b)^2}\lesssim \frac{K^3h^{(t)} }{n^3p_{max}}\vee \frac{K ^3\log K}{n^3p_{max}^2}.\]
By consequence \[ \norm{L_3} \lesssim \frac{1}{\sqrt{I(a,b)}}\frac{Kh^{(t)}}{n} \vee  \frac{\log K}{\sqrt{I(a,b)}np_{max}}\]

\paragraph{Control of $L_4$.} It can be handled in the same way as $L_2$.

We can conclude by summing all the error terms.
\end{proof}
\paragraph{Control of $F_1$.} This term can be controlled with a similar argument as in \cite{csbm_braun22a}. We have \begin{align*}
    \frac{F_1}{\Delta^2(z_i,b)\lt} &\leq \sum_i \norm{E_{i:}(Z^{(t)}-Z)}^2\frac{1}{\Delta_{min}^2\lt}\\
    &\leq  \norm{E_{i:}(Z^{(t)}-Z)}_F^2\frac{1}{\Delta_{min}^2\lt}\\
    &\lesssim K\norm{E}^2 \frac{K}{n\Delta_{min}^6} l(z^{(t)},z) \tag{by Lemma \ref{lem:fond_ineq}}\\
    &\lesssim K\frac{np_{max}}{\Delta_{min}^4} \frac{K\lt}{n\Delta_{min}^2} \to 0
\end{align*}

\paragraph{Control of $F_2$.} This term can be handled again by the same techniques as in \cite{csbm_braun22a}. We have by triangular inequality \begin{align*}
     \frac{F_2}{\Delta^2(z_i,b)\lt} &\leq 4 \sum_i \norm{E_{i:}Z^{(t)}}^2\max_k\norm{\log\Pi_{k:}^{(t)}-\log \tilde{\Pi}_{k:}}^2\frac{1}{\Delta_{min}^2\lt}\\
     &\lesssim K^2\frac{n^2p_{max}\lt}{n^2\Delta_{min}^6} \to 0 \tag{by Lemma \ref{lem:fond_ineq}}.
\end{align*}

\paragraph{Control of $F_3$ while $\sqrt{\frac{h^{(t)}}{n}}>\frac{\sqrt{log K}}{np_{max}}$.} To control $F_3$ we will use the following lemma.
\begin{lemma}\label{lem:norm_noise}
    Under the assumption of Theorem \ref{thm:main} we have with probability at least $1-n^{-\Omega(1)}$ \begin{enumerate}
        \item $\max_{a,b\in [K]}\norm{\sum_{i\in \calC_a,j\in \calC_b}A_{ij} \epsilon_{ij}\epsilon_{ij}^\top}\lesssim (n/K)^2p_{max}$,
        \item $\max_{l,  l', a \in [K]}\norm{\sum_{i\in \calC_a,j\in \calC_l,j'\in \calC_{l'}}A_{ij}A_{ij'}  \epsilon_{ij}\epsilon_{ij'}^\top}\lesssim (n/K)^2p_{max}$.
    \end{enumerate}
\end{lemma}
\begin{proof}The first point can be obtained by conditioning on $A$ and using the Lemma A.2 in \cite{Lu2016StatisticalAC}. Since $\sum_{i\in \calC_a,j\in \calC_b}A_{ij}\lesssim (n/K)^2p_{max}$ w.h.p. we obtain the stated result. The second result can be obtained by a similar argument and by noticing that  by Lemma \ref{lem:conc_quadr} w.h.p \[ \max_{a,l,l'}\sum_{i\in \calC_a,j\in \calC_l,j'\in \calC_{l'}}A_{ij}A_{ij'}\lesssim (n/K)^2p_{max} \]

\end{proof}

By developing the square in $F_3$ we obtain \begin{align*}
    \frac{F_3}{\Delta^2(z_i,b)\lt} &\lesssim \frac{K^4}{\Delta_{min}^2\lt} \left(\norm{\sum_{i,j}A_{ij} \epsilon_{ij}\epsilon_{ij}^\top}+\max_{l, l', a \in [K]}\norm{\sum_{i\in \calC_a,j\in \calC_l,j'\in \calC_{l'}}A_{ij}A_{ij'}  \epsilon_{ij}\epsilon_{ij'}^\top}\right)\\
    &\times \max_{b\neq a, l}\norm{ \mu_{al}^{(t)}-\tilde{\mu}_{al}+\tilde{\mu}_{bl}-\mu_{bl}^{(t)}}^2\\
    &\lesssim K^6\frac{np_{max}}{\Delta_{min}^4} \to 0 \tag{by Lemma \ref{lem:fond_ineq}}.
\end{align*}

\paragraph{Control of $F_4$ while $\sqrt{\frac{h^{(t)}}{n}}>\frac{\sqrt{log K}}{np_{max}}$.}

Let us write $c_{abl}=\norm{\tilde{\mu}_{bl}-\mu_{bl}^{(t)}}^2-2\langle \mu_{al}-\mu_{bl}, \tilde{\mu}_{bl}-\mu_{bl}^{(t)}\rangle $. We have $\max_{a,b,l}|c_{abl}|\lesssim K^2\sqrt{\frac{\lt}{n\Delta_{min}^2}}$
By developing the square in $F_4$ we obtain
\begin{align*}
    \frac{F_4}{\Delta^2(z_i,b)\lt}&\lesssim \sum_i \sum_{b\neq z(i),l,l'}\sum_{j\in \calC_l, j'\in \calC_{l'}}E_{ij}E_{ij'}c_{z(i)bl}c_{z(i)bl'}\frac{1}{\Delta_{min}^2\lt}\\
    &\lesssim \sum_{b\neq z(i),l,l', l''}c_{l''bl}c_{l''bl'}\sum_{j \in \calC_l, j'\in \calC_{l'}}  \sum_{i\in \calC_{l''}}E_{ij} E_{ij'}\\
    &\lesssim K^3\frac{n^2p_{max}}{n\Delta^4_{min}}\to 0. \tag{by Lemma \ref{lem:conc_quadr}}
\end{align*}

\paragraph{Case where $\sqrt{\frac{h^{(t)}}{n}}<\frac{\sqrt{log K}}{np_{max}}$.} In this case, one should consider $F_3$ and $F_4$ as G-error terms.  More precisely, we should consider the following term appearing in the F-error decomposition \[\sum_{l\in [K]}\sum_{j\in \calC_l}\left(2A_{ij}\langle \epsilon_{ij}, \mu_{al}^{(t)}-\tilde{\mu}_{al}+\tilde{\mu}_{bl}-\mu_{bl}^{(t)}\rangle + E_{ij}\norm{\tilde{\mu}_{bl}-\mu_{bl}^{(t)}}^2-2E_{ij}\langle \mu_{al}-\mu_{bl}, \tilde{\mu}_{bl}-\mu_{bl}^{(t)}\rangle \right)\] as a G-error term.
By using the fact that \[\norm{\tilde{\mu}_{bl}-\mu_{bl}^{(t)}} \lesssim \frac{K^{1.5}\log K}{np_{max}}\] it is easy to show that all the terms are $o(1)$. See also the proof of the bound of $G_7$ and $G_8$ in the next subsection.

\subsubsection{G-error term}\label{app:g_err}

We can upper bound $G_i^{(t)}$ by $G_1+G_2+G_3+G_4+G_5+G_6+G_7+G_8$ where \begin{align*}
    G_1&= \abs{\langle P_{i:}(Z^{(t)}-Z), \log \Pi_{a:}- \log \Pi_{b:}\rangle}\\
    G_2&=  \abs{\langle P_{i:}Z^{(t)}, \log \Pi_{a:}^{(t)}- \log \Pi_{a:}+\log\Pi_{b:}^{(t)}+ \log \Pi_{b:}\rangle}\\
    G_3&=\sum_{l\in [K]}\sum_{j\in \calC_l}P_{ij}\norm{\mu_{bl}-\mu_{bl}^{(t)}}^2\\
    G_4&=\sum_{l\in [K]}\sum_{j\in \calC_l}2P_{ij}\abs{\langle \mu_{al}-\mu_{bl}, \mu_{bl}-\mu_{bl}^{(t)}\rangle}\\
    G_5&=\abs{\langle E_{i:}Z^{(t)}, \log \tilde{\Pi}_{a:}-\log \Pi_{a:}+\log \Pi_{b:}-\log \tilde{\Pi}_{b:}\rangle}\\
    G_6&=2\abs{\sum_{l\in [K]}\sum_{j\in \calC_l}E_{ij}\langle \mu_{al}-\mu_{bl}, \mu_{bl}-\tilde{\mu}_{bl}\rangle}\\
    G_7&= \abs{\sum_{l\in [K]}\sum_{j\in \calC_l}E_{ij}\left(\norm{\mu_{bl}-\mu_{bl}^{(t)}}^2-\norm{\tilde{\mu}_{bl}-\mu_{bl}^{(t)}}^2\right)}\\
    G_8&=2\abs{ \sum_{l\in [K]}\sum_{j\in \calC_l}A_{ij}\langle \epsilon_{ij}, \tilde{\mu}_{al}-\mu_{al}-\tilde{\mu}_{bl}+\mu_{bl}\rangle}.
\end{align*}

\paragraph{Control of $G_1$.}
Since $\norm{P_{i:}(Z^{(t)}-Z)}\leq \sqrt{K}p_{max}h^{(t)}$ and $\norm{\log \Pi_{a:}- \log \Pi_{b:}}= O(1)$ we get \[ \frac{G_1}{\Delta^2(z_i,b)}\lesssim \frac{\sqrt{K}p_{max}h^{(t)}}{\Delta_{min}^2}\lesssim \frac{K\lt}{n\Delta_{min}^2} \frac{np_{max}}{\Delta_{min}^2\sqrt{K}}\leq \delta.\]

\paragraph{Control of $G_2$.}
We have $\norm{P_{i:}Z^{(t)}}\lesssim (n/\sqrt{K})p_{max}$ and by Lemma \ref{lem:fond_ineq} $\norm{\log \Pi_{a:}^{(t)}- \log \Pi_{a:}}\lesssim K\frac{\lt}{n\Delta_{min}^2}$. By using the triangular inequality, we obtain \[ \frac{G_2}{\Delta^2(z_i,b)} \lesssim K\frac{\lt}{n\Delta_{min}^2} \frac{np_{max}}{\Delta_{min}^2\sqrt{K}} \leq \delta. \]

\paragraph{Control of $G_3$.}
By Lemma \ref{lem:fond_ineq} $ \max_{b, l\in [K]} \norm{\mu_{bl}-\mu_{bl}^{(t)}} \lesssim K^2\sqrt{\frac{\lt}{n\Delta^2_{min}} }+\frac{K}{\sqrt{n}} $. By consequence, \[\frac{G_3}{\Delta^2(z_i,b)} \lesssim K^4\frac{np_{max}}{\Delta_{min}^2}\left(\frac{\lt}{n\Delta^2_{min}}+\frac{1}{n}\right)\leq \delta.\]

\paragraph{Control of $G_4$.}
The proof is similar to $G_3$.

\paragraph{Control of $G_5$.}
One has $\norm{\log \tilde{\Pi}_{a:}-\log \Pi_{a:}+\log \Pi_{b:}-\log \tilde{\Pi}_{b:}}\lesssim \frac{1}{n^2p_{max}}$ and $\max_i\norm{E_{i:}Z^{(t)}}\lesssim \sqrt{np_{max}}$.

\paragraph{Control of $G_6$.} Let $V\in \R^k$ such that $V_l=\langle \mu_{al}-\mu_{bl}, \mu_{bl}-\tilde{\mu}_{bl}\rangle$. By the triangular inequality and Lemma \ref{lem:fond_ineq}, we have $\norm{V}_\infty \lesssim \frac{K}{n\sqrt{p}}$.
Hence \begin{align*}
    \frac{G_6}{\Delta^2(z_i,b)} &= \frac{2\abs{\langle E_{i:}Z, V \rangle}}{\Delta^2(z_i,b)} \\
    &\lesssim \frac{\sqrt{n}p_{max}\norm{V}_\infty}{\Delta_{min}^2} = o(1)\\ 
\end{align*}

\paragraph{Control of $G_7$.} By Lemma \ref{lem:fond_ineq}, $ \max_{b, l\in [K]} \norm{\mu_{bl}-\mu_{bl}^{(t)}} \lesssim K^2\sqrt{\frac{\lt}{n\Delta^2_{min}} }+\frac{K}{\sqrt{n}} $ and  $ \max_{b, l\in [K]} \norm{\tilde{\mu}_{bl}-\mu_{bl}^{(t)}} \lesssim K\sqrt{\frac{\lt}{n\Delta^2_{min}} }$. By consequence, \begin{align*}
    \frac{G_6}{\Delta^2(z_i,b)} & \lesssim \frac{\sqrt{np_{max}}}{\Delta^2_{min}}\left(K^4\frac{\lt}{n\Delta^2_{min}} + K^2\frac{1}{n}\right)= o(1).
\end{align*}

\paragraph{Control of $G_8$.} Conditionally on $(A_{ij})_{j\in \calC_l}$, $\sum_{j\in \calC_l}A_{ij}\langle \epsilon_{ij}, \tilde{\mu}_{al}-\mu_{al}-\tilde{\mu}_{bl}+\mu_{bl}\rangle$ is a centered Gaussian r.v. with variance $\sigma_A^2 = \norm{\tilde{\mu}_{al}-\mu_{al}-\tilde{\mu}_{bl}+\mu_{bl}}^2 \sum_{j\in \calC_l} A_{ij} \lesssim \frac{K}{\sqrt{n}}\sum_{j\in \calC_l} A_{ij} $ by Lemma \ref{lem:fond_ineq}. By consequence,
\begin{align*}
    \prob_{(A_{ij})_{j\in \calC_l}}\left(\abs{\sum_{j\in \calC_l}A_{ij}\langle \epsilon_{ij}, \tilde{\mu}_{al}-\mu_{al}-\tilde{\mu}_{bl}+\mu_{bl}\rangle} \gtrsim \sqrt{\log n}\norm{\tilde{\mu}_{al}-\mu_{al}-\tilde{\mu}_{bl}+\mu_{bl}}^2\sum_{j\in \calC_l} A_{ij} \right) \leq n^{-\Omega(1)}.
\end{align*}
Let us denote the event  \[ H= \left\lbrace \max_{i,l} \sum_{j\in \calC_l} A_{ij}\lesssim \frac{np_{max}}{K} \right\rbrace.\] This event holds with probability at least $1-n^{-\Omega(1)}$.  We have 
\begin{align*}
    &\prob \left( \abs{\sum_{j\in \calC_l}A_{ij}\langle \epsilon_{ij}, \tilde{\mu}_{al}-\mu_{al}-\tilde{\mu}_{bl}+\mu_{bl}\rangle} \gtrsim \sqrt{\log n}\frac{np_{max}}{K}\norm{\tilde{\mu}_{al}-\mu_{al}-\tilde{\mu}_{bl}+\mu_{bl}}^2\right)\\ 
    &\leq \prob_H\left(\abs{\sum_{j\in \calC_l}A_{ij}\langle \epsilon_{ij}, \tilde{\mu}_{al}-\mu_{al}-\tilde{\mu}_{bl}+\mu_{bl}\rangle} \gtrsim \sqrt{\log n}\frac{np_{max}}{K}\norm{\tilde{\mu}_{al}-\mu_{al}-\tilde{\mu}_{bl}+\mu_{bl}}^2 \right) + n^{-\Omega(1)} \\  
    &\lesssim \prob_H\left( \prob_{(A_{ij})_{j\in \calC_l}}\left(\abs{\sum_{j\in \calC_l}A_{ij}\langle \epsilon_{ij}, \tilde{\mu}_{al}-\mu_{al}-\tilde{\mu}_{bl}+\mu_{bl}\rangle} \gtrsim \sqrt{\log n}\norm{\tilde{\mu}_{al}-\mu_{al}-\tilde{\mu}_{bl}+\mu_{bl}}^2\sum_{j\in \calC_l} A_{ij} \right) \right)+ n^{-\Omega(1)}\\
    &\leq n^{-\Omega(1)}.
\end{align*}
Since $\norm{\tilde{\mu}_{al}-\mu_{al}-\tilde{\mu}_{bl}+\mu_{bl}}^2\lesssim \frac{K^2}{n}$ by Lemma \ref{lem:fond_ineq}, we obtain that \[ \frac{G_8}{\Delta^2(z_i,b)} \lesssim \frac{K^3\sqrt{n}np_{max}}{n\Delta_{min}^2}=o(1).\]

\section{Proof of Theorem \ref{thm:minimax}}\label{app:minimax}
    Choose two communities $a$ and $b\in [K]$ such that $\Delta_{\min} = \Delta (a,b)$.
    For each $k \in [K]$, let $T_k$ a subset of $\calC_k$ with cardinality $\frac{3n}{4K}$. Define $T=\cup_{k=1}^KT_k$ and \[ \mathcal{Z} = \lbrace \hat{z} : \hat{z}_i=z_i \text{ for all }i\in T \rbrace.\]
    By using the same argument as in the proof of Theorem 2 in \cite{gaodcsbm} we can reduce the problem to a two-hypothesis testing problem
    \begin{equation}\label{eq:minmax1}
    \inf_{\hat{z}} \sup_{\theta \in \Theta} \expec (r(\hat{z},z)) \geq \frac{1}{6|T^c|}\sum_{i \in T^c}\frac{1}{2K^2} \inf_{\hat{z}_i} \prob_1(\hat{z}_i=2)+\prob_2(\hat{z}_i=1)
    \end{equation}  
    where $\prob_1$ (resp. $\prob_2$ ) denotes the probability distribution of the data when $z_i=a$ (resp. $z_i=b$). By the Neyman Pearson Lemma, the likelihood ratio test achieves the infinimum of the right-hand side of \eqref{eq:minmax1}. Hence we have \[ \inf_{\hat{z}_i} \prob_1(\hat{z}_i=2)+\prob_2(\hat{z}_i=1) = \prob\left(\underbrace{\sum_{l \in [K], j\in \calC_l}A_{ij}(\norm{\mu_{al}-\mu_{bl}}^2+2\langle \epsilon_{ij}, \mu_{al}-\mu_{al}\rangle) \leq 0 }_{\calO}\right)  .\]

First, assume that $\Delta^2(a,b)\to \infty$.
 Conditionally on $(A_{ij})_j$, $\prob_{(A_{ij})_j}(\calO)= \prob(X_A\leq -\frac{\sigma_A^2}{2})=\prob(X_A\geq \frac{\sigma_A^2}{2}) $ where $X_A\sim \calN(0, \sigma^2_A)$ and $\sigma^2_A=\sum_{l,j}A_{ij}\norm{\mu_{al}-\mu_{bl}}^2$. 
 By using the fact that \[ \int_t^\infty e^{-x^2/2}\diff x \geq \frac{1}{\sqrt{2\pi}}\frac{t}{t^2+1}e^{-t^2/2},\] it is easy to show that \[ \prob\left(X_A\geq \frac{\sigma_A^2}{2}\right) \gtrsim \frac{e^{-\sigma_A^2/8}}{\sigma_A}. \] By Chernoff's bound, there exists constants $c_1, c_2>1$ such that \[ \prob\left(\frac{1}{c_1}\expec(\sigma_A^2)\leq \sigma_A^2\leq c_1\expec(\sigma_A^2)\right) \geq 1-e^{-c_2\expec(\sigma_A^2)}.\] Moreover we have 
 \begin{align*}
     \log \expec\left(e^{-\sigma_A^2/8}\right) &= \sum_{l,j} \log \left(e^{-\norm{\mu_{al}-\mu_{bl}}^2/8}p+1-p\right) \\
     &\geq (1-o(1)) \sum_{l,j} p\left(e^{-\norm{\mu_{al}-\mu_{bl}}^2/8}-1\right) \tag{by using $\log (1+x)\geq \frac{x}{1+x}$ for $x>-1$}\\
     &\geq -(1-o(1))\frac{\Delta^2(a,b)}{8} \tag{because $e^{-x}-1\geq -x$}.
 \end{align*}
 By consequence, 
 \begin{align*}
     \prob(\calO) &= \expec\left(\prob\left(X_A\geq \frac{\sigma_A^2}{2}\right) \right)\\
     &\geq \expec\left(\frac{e^{-\sigma_A^2/8}}{\sigma_A} \indic_{\lbrace \frac{1}{c_1}\expec(\sigma_A^2)\leq \sigma_A^2\leq c_1\expec(\sigma_A^2)\rbrace } \right)\\
     &\geq \frac{1}{\sqrt{c_1\expec(\sigma_A^2)}}\expec\left(e^{-\sigma_A^2/8} \indic_{\lbrace \frac{1}{c_1}\expec(\sigma_A^2)\leq \sigma_A^2\leq c_1\expec(\sigma_A^2)\rbrace } \right)\\
     &\geq \frac{1}{\sqrt{c_1}\Delta(a,b)}\expec\left(e^{-\sigma_A^2/8} \right)-\frac{e^{-c_2\Delta^2(a,b)}}{\sqrt{c_1}\Delta(a,b)}\\
     &\geq e^{-(1-o(1))\Delta^2(a,b)}
 \end{align*} since $c_2>1$ and $\Delta(a,b)\to \infty$.

 If $\Delta^2(a,b)=O(1)$, then $\sigma_A^2=O(1)$ with a positive probability and \[ \prob(\calO) \geq \prob\left(\sigma_A^2=O(1)\right)\expec\left(\prob\left(X_A\geq \frac{\sigma_A^2}{2}\right)\right) \gtrsim 1.\]

 \section{Proof of Lemma \ref{lem:conc_quadr}}\label{app:lem3}
\paragraph{Control of $\calE$.}The event $\calD$ occurs with probability at least $1-n^{-5}$ by Chernoff's bound. It remains to show that conditionally on $\calD$ \[ \max_\Lambda \sum_{i\in \calC_k}\left(\sum_{j \in \Lambda_\delta^c} E_{ij}\right)^2 \lesssim (n/K)^2p_{max}.\]
Let us fix $\Lambda$ and define $X_i=(\sum_j E_{ij})^2\indic_{\lbrace \abs{\sum_j E_{ij} }\lesssim np_{max}/K \rbrace }$. By developing the square and using the independence between $E_{ij}$ and $E_{ij'}$ for $j\neq j'$ we obtain $\expec(X_i)\lesssim np_{max}$. A similar calculation shows that $\Var(X_i)\lesssim (np)^2$. Hence, Bernstein's inequality gives \[ \prob\left(\sum_i X_i\gtrsim (n/K)^2p_{max}\right)\leq e^{-\Omega(n/K)}.\] We obtain the result by a union bound and the fact that conditionally on $\calD$, it holds that $\max_i\abs{\sum_j E_{ij} }\lesssim np_{max}/K$.
 
 \paragraph{Case where $k=k'=k''$.} 
 
 One can first decouple the indexes $i$ from $j, j'$ by considering  \[ S_\delta = \sum_{i, j\neq j'}\delta_i(1-\delta_j)(1-\delta_{j'})E_{ij}E_{ij'} = \sum_{i\in \Lambda} \sum_{j, j'\in \Lambda^c}E_{ij}E_{ij'}.\] 
 Then one can use the result from the case  $k\neq k'=k''$ to show that conditionally on $\calE$, $S_\delta\lesssim (n/K)^2p_{max}$ with probability at least $1-e^{-Cn/K}$ for some constant $C>1$ and conclude as in the previous case.